\documentclass[journal]{IEEEtran}
\usepackage{times}

\usepackage[numbers,sort,compress]{natbib}
\usepackage{multicol}
\usepackage[bookmarks=true]{hyperref}
\usepackage{comment}
\usepackage{amsthm}
\usepackage{mathtools}
\usepackage{amsmath}
\usepackage{amsfonts}
\usepackage{soul} 
\usepackage{amsmath,amsfonts,amssymb}
\usepackage[mathscr]{euscript}
\usepackage{hyperref}
\usepackage{graphicx}
\usepackage[dvipsnames]{xcolor} 

\usepackage{algorithm}
\usepackage[noend]{algpseudocode}
\algrenewcommand{\algorithmicrequire}{\textbf{Input:}}
\algrenewcommand{\algorithmicensure}{\textbf{Output:}}

\hypersetup{
    colorlinks=true,
    linkcolor=black,
    citecolor=blue,
    urlcolor=black
}

\newtheorem{remark}{Remark}[section]
\newtheorem{theorem}{Theorem}[section]
\newtheorem{problem}{Problem}[section]
\newtheorem{proposition}{Proposition}[section]
\newtheorem{lemma}{Lemma}[section]
\newtheorem{definition}{Definition}[section]
\newcounter{subproblem}[problem] 

\usepackage{xcolor}

\renewcommand{\bf}[1]{\textbf{#1}}
\newcommand{\mycomment}[1]{}
\renewcommand{\t}[0]{^{\intercal}}
\newcommand{\inv}[0]{^{-1}}
\begin{document}

\title{SDP Synthesis of Distributionally Robust Backward Reachable Trees for Probabilistic Planning}

\author{Naman Aggarwal
 and Jonathan P. How, \IEEEmembership{IEEE Fellow} 
\thanks{This work was supported by Lockheed Martin Corporation. The authors are with 
Laboratory for Information and Decision Systems, Massachusetts Institute of Technology, Cambridge MA 02139. Email: \texttt{\{namanagg, jhow\}}@mit.edu.
}}



%

\maketitle

\begin{abstract}
    The paper presents Maximal Ellipsoid Backward Reachable Trees MAXELLIPSOID BRT, which is a multi-query algorithm for planning of dynamic systems under stochastic motion uncertainty and constraints on the control input. In contrast to existing probabilistic planning methods that grow a roadmap of distributions, our proposed method introduces a framework to construct a roadmap of ambiguity sets of distributions such that each edge in our proposed roadmap provides a feasible control sequence for a family of distributions at once leading to efficient multi-query planning. Specifically, we construct a backward reachable tree of maximal size ambiguity sets and the corresponding distributionally robust edge controllers. Experiments show that the computation of these sets of distributions, in a backwards fashion from the goal, leads to efficient planning at a fraction of the size of the roadmap required for state-of-the-art methods. The computation of these maximal ambiguity sets and edges is carried out via a convex semidefinite relaxation to a novel nonlinear program. We also formally prove a theorem on maximum coverage for a technique proposed in our prior work \cite{aggarwal2024sdp_CDC}.
\end{abstract}

\IEEEpeerreviewmaketitle
\section{Introduction}
Multi-query motion planning entails the design of plans that can be reused across different initial configurations of the system. This is typically done via the offline construction of a roadmap of feasible trajectories in the state space, such that in real-time, for a pair of initial and goal configurations of the system, planning proceeds by connecting the initial and goal configurations to the roadmap followed by graph search to find a path \cite{prm1996,eirm}. The pre-computation of a roadmap is beneficial since it avoids the computational burden of finding plans from scratch for new configurations, and reuses search effort across queries. An important design consideration for any roadmap-based planning algorithm is the coverage of the roadmap -- it is desirable to be able to re-use search effort across as many queries as possible, therefore, reasoning explicitly about the coverage of the roadmap becomes an important algorithmic design consideration.

There is extensive work on motion planning via the construction of reusable roadmaps for deterministic systems \cite{tedrake2010lqr,majumdar2017funnel,pipx}. The work proposed in \cite{tedrake2010lqr}, for instance, builds a tree of \textit{funnels}, which are regions of finite-time invariance around trajectories with associated feedback controllers, backwards from the goal in a space-filling manner. Paths to the goal could be found from different regions of the state space by traversing the branches of the constructed tree. 

\begin{figure}[t]
\vspace*{-0.15in}
    \centering
    \includegraphics[width=0.5\textwidth]{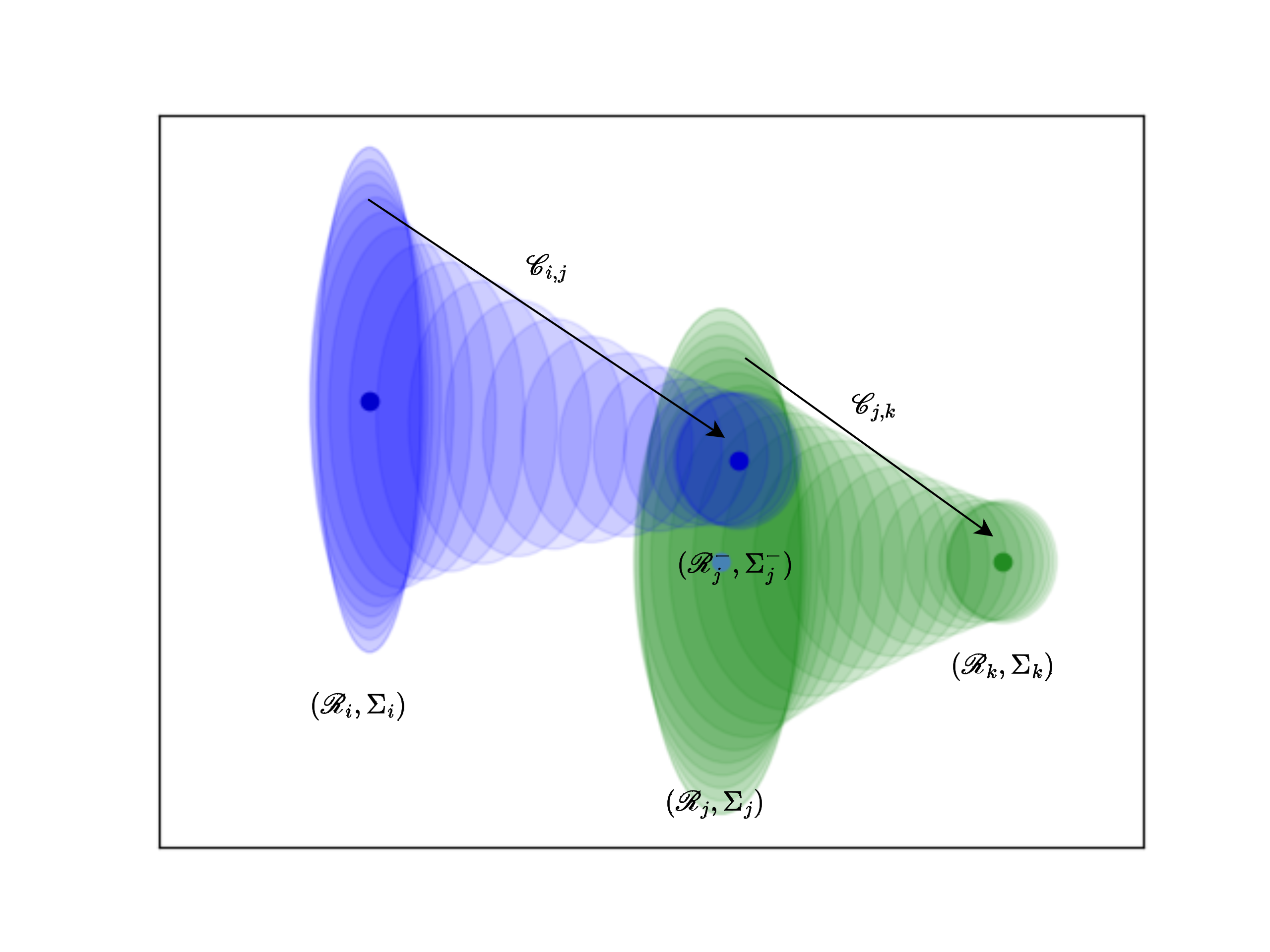}
\vspace*{-0.5in}
    \caption{Recursive feasibility through sequential composition (refer to Lemma~\ref{lemma:sequential_composition} for a formalization). An illustration of the satisfaction of constraints (goal reaching constraint, and chance constraints on the state and control input) for a $2N$-step trajectory initialized at the $(\mu_{i}, \Sigma_{i})$ distribution driven through a sequential composition of two $N$-step controllers, $ \mathscr{C}_{i,j}$ designed for the $(\mathscr{R}_{i}, \Sigma_{i}) \rightarrow (\mathscr{R}^{-}_{j}, \Sigma^{-}_{j})$ maneuver, and $\mathscr{C}_{j,k}$ designed for the $(\mathscr{R}_{j}, \Sigma_{j}) \rightarrow (\mathscr{R}_{k}, \Sigma_{k})$ maneuver where $ \mathscr{R}^{-}_{j} \subset \mathscr{R}_{j} $ and $ \Sigma^{-}_{j} \prec \Sigma_{j} $. The distribution $(\mu_{i}, \Sigma_{i})$ is steered to the goal $(\mathscr{R}_{k}, \Sigma_{k})$ in $2N$-steps under the concatenated controller such that all chance constraints on the state and control input are satisfied for the $2N$-step trajectory.}
    \label{fig:intro_picture}
\end{figure}
In this paper, we are concerned with multi-query planning for systems with stochastic dynamics via roadmap based methods under constraints on the control input. FIRM \cite{agha2014firm} is one multi-query planning framework that builds a roadmap in the belief space by sampling stationary belief nodes and drawing belief stabilizing controllers as edges between the nodes. The closest in approach to our work is the recently proposed CS-BRM \cite{csbrm-tro} that provides a faster planning scheme compared to FIRM by allowing sampling of non-stationary belief nodes and adding finite-time feedback controllers as edges. CS-BRM extends recent work on covariance steering that allows finite-time satisfaction of terminal covariance constraints \cite{okamoto2018optimal,liu2022optimal,rapakoulias2023discrete}, but it does not impose control input constraints and proceeds by adding edges between randomly sampled belief nodes. In the absence of control constraints, the distribution steering problem decouples into the mean steering and covariance steering problems: both of which are tractable to solve \cite{rapakoulias2023discrete} such that there always exists a feasible control sequence for a steering maneuver defined by an arbitrary pair of belief nodes. 
Hence, the question of coverage is trivial in the absence of control constraints since the roadmap could be reached from arbitrary distributions or, alternatively, arbitrary distributions could be steered to the roadmap. 

In the presence of control input constraints however, the reachability of belief nodes has to be explicitly established. Even for the simplest case of a controllable linear-Gaussian system, the existence of a steering maneuver for a pair of belief nodes must be determined by checking the feasibility of a nonconvex program, and it is no longer possible to drive arbitrary distributions to arbitrary distributions \cite{liu2022optimal} \cite{aggarwal2024sdp}. Therefore, in stochastic planning settings with constraints on the control input, it is desirable to construct a roadmap that finds paths from the largest possible set of query distributions. Existing methods in the multi-query stochastic planning literature populate the roadmap by first sampling a candidate mean, followed by a procedure that determines a candidate covariance corresponding to the belief node, either through sampling \cite{csbrm_tro2024} \cite{agha2014firm}, or through some computation \cite{aggarwal2024sdp}. In approaches that sample candidate belief nodes randomly, both the mean and the covariance \cite{csbrm_tro2024} \cite{agha2014firm}, there is no explicit consideration of coverage while adding nodes nor a post-analysis on the coverage characteristics of the constructed roadmap. 

Our prior work \cite{aggarwal2024sdp_CDC} formally characterizes the notion of the coverage of a roadmap in a stochastic domain by introducing $h$-hop backward reachable sets of distributions ($h$-BRS). The $h$-BRS of a distribution $p$ is defined as the set of all distributions for which there exists a control sequence that can steer the system to $p$, and is mathematically defined through the feasibility of a (nonconvex) program corresponding to the steering maneuver. Ref. \cite{aggarwal2024sdp_CDC} also introduces a novel procedure for adding nodes and edges through the solution of an optimization program that provides provably maximum coverage compared to other methods that grow a roadmap of distributions as demonstrated in \cite{aggarwal2024sdp_CDC} and \textit{as proven in this paper}. The optimization program to add belief nodes and edges is designed with the philosophy of maximum reuse of computation such that the edge controller provides a feasible steering sequence for the largest possible size of the set of initial conditions (by maximizing the minimum eigenvalue of the covariance corresponding to the candidate belief node) \cite{aggarwal2024sdp_CDC}.

In this paper, we develop a novel multi-query stochastic planning algorithm that grows a roadmap of families (henceforth referred to as an ambiguity set \cite{dro_lin2022distributionally, dro_Rahimian_2022}) of distributions, which contrasts  with existing methods in the literature \cite{agha2014firm, csbrm_tro2024, aggarwal2024sdp_CDC} that grow a roadmap of distributions. Specifically, we grow a backward reachable tree of maximal size ambiguity sets and the corresponding maximally distributionally robust edge controllers. The computation of these maximal size ambiguity sets and the associated edge controllers in a backwards fashion from the goal leads to a compact roadmap representation that is efficient to compute and construct. This is due to the general notion of \textbf{reuse} (see Remark~\ref{remark:reuse}) that provides a feasible control sequence for a family of distributions at once (the ambiguity set) leading to efficient multi-query planning at a fraction of the size of the roadmaps required for state-of-the-art methods (see Section~\ref{sec:experiments} for experiments). Each node of our proposed roadmap-based algorithm represents an ambiguity set of distributions such that a directed edge between a pair of nodes represents a finite-time distributionally robust covariance steering controller corresponding to the respective ambiguity sets. For our formulation, the ambiguity sets are parameterized as an ambiguity set for the first-order moment, and a prescribed second-order moment of the distribution. Our proposed framework constructs the roadmap by adding nodes and computing the corresponding edge controllers such that the size of the ambiguity set corresponding to the candidate node is maximized. This allows maximum reuse of computation since the edge controller provides a feasible control sequence for any distribution in the ambiguity set at the node with the outgoing edge.
The main contributions are as follows.
\begin{itemize}
    \item We introduce a new roadmap-based algorithm for multi-query planning in stochastic domains that grows backward reachable trees of maximal size ambiguity sets of distributions and the corresponding edge controllers and is efficient to compute and construct.
    \item We develop a novel method for node addition by computing the maximal backward reachable ambiguity set and the corresponding maximally distributionally robust edge controller using a nonlinear program. We then provide a convex SDP relaxation that recovers a feasible solution to the nonlinear program.
    \item We demonstrate through extensive simulations the efficacy of our approach on a 6 DoF model that shows superior performance compared to state-of-the-art methods \cite{aggarwal2024sdp_CDC, csbrm_tro2024} on an example planning application.
    \item We provide a formal proof for the maximum coverage theorem of the roadmap technique proposed in \cite{aggarwal2024sdp_CDC}.
\end{itemize}
\textit{Related work on distributionally robust optimal control (DRO-OC):} There is existing work on distributionally robust controller design that focuses on computing optimal control sequences between pairs of ambiguity sets of distributions \cite{yang2020wasserstein,shapiro2021distributionally, renganathan2023distributionally, pan2023distributionally, li2021distributionally}. We distinguish ourselves in that we are concerned with constructing backward reachable trees through the computation of maximal size ambiguity sets and the associated maximally distributionally robust edge controllers to facilitate multi-query planning.

\section{Problem Statement}\label{sec:problem_statement}
\subsection{Notation and Mathematical Preliminaries}
We denote the set of all positive semi-definite (positive definite) matrices in $\mathbb{R}^{n \times n}$ as $\mathbb{S}^{n}_{+}$ ($\mathbb{S}^{n}_{++}$). For a vector $c \in \mathbb{R}^{n}$ and a matrix $\mathscr{P} \in \mathbb{S}^{n}_{++}$, $\mathscr{R}(c, \mathscr{P})$ denotes an $n$-dimensional ellipsoid with center $c$ and shape matrix $\mathscr{P}$, such that, $ \mathscr{R}(c, \mathscr{P}) \coloneqq \{ x \in \mathbb{R}^{n} \ \vert \ (x - c)\t \mathscr{P}\inv (x-c) \leq 1 \} $.

We assume that all random objects are defined on a common underlying probability space $ (\Omega, \mathcal{F}, \mathbb{P})$. An $l$-dimensional real random variable $ \mathbf{x} $ is defined as a measurable function from the sample space to $ \mathbb{R}^{l} $, $ \mathbf{x}: \Omega \rightarrow \mathbb{R}^{l}$, such that $\mathbb{P}_{\mathbf{x}}$ denotes the measure induced by $ \mathbf{x} $ on $ \mathbb{R}^{l} $. $ \mathbb{P}_{\mathbf{x}} \coloneqq \mathbb{P}( \mathbf{x}\inv ( \mathcal{B} (\mathbb{R}^{l}) ))$ where $ \mathcal{B}(\mathbb{R}^{l})$ denotes the Borel $\sigma-$algebra on the $l$-dimensional Euclidean space. Note that $ \mathbb{P}_{\mathbf{x}} \coloneqq \mathbb{P}(\mathbf{x}\inv( \mathcal{B}(\mathbb{R}^{l}) )) $ is a shorthand for $ \mathbb{P}_{ \mathbf{x} }(A) \coloneqq \mathbb{P} ( \mathbf{x}\inv (A) ) \ \forall \ A \subset \mathcal{B}( \mathbb{R}^{l} )$. For two random variables $\mathbf{x}_{0}$ and $\mathbf{x}_{1}$, $ l_{0} $ and $l_{1}$ dimensional respectively, we denote their joint probability measure by $ \mathbb{P}_{\mathbf{x}_{0}, \mathbf{x}_{1}}$.

We use the notation $ \mathcal{P} (\mathcal{W}) $ to define the set of all probability distributions on any measurable Borel space $\mathcal{W}$. Also, we use the notation $ \mathcal{P}^{\mathbf{x}} $ to define the ambiguity set of distributions for a random variable $\mathbf{x}$ such that $ \mathbb{P}_{\mathbf{x}} \in \mathcal{P}^{\mathbf{x}} \subset \mathcal{P}( \mathcal{B}(\mathbb{R}^{l}) )$. For a pair of random variables $\mathbf{x}_{0}$ and $\mathbf{x}_{1}$, $l_{0}$ and $l_{1}$-dimensional respectively, we denote the ambiguity set of their joint distribution as $ \mathcal{P}^{\mathbf{x}_{0}, \mathbf{x}_{1}} $ such that $ \mathbb{P}_{\mathbf{x}_{0}, \mathbf{x}_{1}} \in \mathcal{P}^{\mathbf{x}_{0}, \mathbf{x}_{1}} \subset \mathcal{P}( \mathcal{B}( \mathbb{R}^{l_{0} + l_{1}} ) )$.

Finally, we denote an $l$-dimensional Gaussian distribution with mean $\mu$ and covariance $\Sigma$ as $\mathcal{N}_{l}(\mu, \Sigma)$.

\subsection{Problem}
Consider a discrete-time stochastic linear system represented by the following difference equation
\begin{equation}\label{eq:dynamics}
    \mathbf{x}_{t+1} = A \mathbf{x}_{t} + B \mathbf{u}_{t} + D \mathbf{w}_{t}
\end{equation}
where $\mathbf{x}_{t} \in \mathcal{X}$, $\mathbf{u}_{t} \in \mathcal{U}$ are the state and control inputs at time $t$, respectively, $\mathcal{X} \subseteq \mathbb{R}^{n}$ is the state space, $\mathcal{U} \subseteq \mathbb{R}^{m}$ is the control space, $n, m \in \mathbb{N}$, $D \in \mathbb{R}^{n \times n}$ and $\mathbf{w}_{t}$ represents the stochastic disturbance at time $t$ that is assumed to have a Gaussian distribution with zero mean and unitary covariance. We assume that $\{w_{t}\}$ is an \textit{i.i.d.} process, and that $A$ is non-singular.

We define a finite-horizon optimal control problem, $\operatorname{OPTSTEER}\left(q, \mathcal{G}, N\right)$, where $q$, $\mathcal{G}$, $N$ are the initial distribution, (ambiguity set corresponding to the) goal distribution, and the time-horizon corresponding to the control problem respectively. $\operatorname{OPTSTEER}(q, \mathcal{G}, N)$ solves for a control sequence that is optimal with respect to a performance index $J$, and that steers the system from an initial distribution $q$ to a goal $\mathcal{G} \coloneqq (\mathscr{R}_\mathcal{G}, \Sigma_{\mathcal{G}})$ in a time-horizon of length $N$ under chance constraints on the state such that the control sequence respects the prescribed constraints on the control input at each time-step:
\begin{equation}
    \min _{ \Phi_{k} } \quad J=\mathbb{E}\left[\sum_{k=0}^{N-1} \mathbf{x}_k^{\top} Q_k \mathbf{x}_k + \mathbf{u}_k^{\top} R_k \mathbf{u}_k\right] \tag{OPTSTEER} \label{eq:J}
\end{equation}
such that, for all $k = 0, 1, \cdots N - 1$,
\begin{subequations}
\begin{align}
    &\mathbf{x}_{k+1} = A\mathbf{x}_{k} + B\mathbf{u}_{k} + D \mathbf{w}_{k}, \label{eq:opt_dynamics_optsteer} \\
    &\mathbf{x}_{0} \sim \mathcal{N}_{n}(\mu_{q}, \Sigma_{q}), \label{eq:opt_init_dist_optsteer} \\
    &\mathbf{x}_{N} \sim \mathcal{N}_{n}(\mu_{N}, \Sigma_{N}), \label{eq:opt_final_dist_optsteer} \\
    &\mu_{N} \in \mathscr{R}_{\mathcal{G}}, \label{eq:opt_goal_mean} \\
    &\lambda_{\mathrm{max}}(\Sigma_{N}) \leq \lambda_{\mathrm{min}}(\Sigma_{\mathcal{G}}), \label{eq:opt_goal_covar} \\
    &\mathbb{P}(\mathbf{x}_{k} \in \mathcal{X}) \geq 1- \epsilon_{x}, \label{eq:opt_state_constraint_optsteer} \\
    &\mathbb{P}(\mathbf{u}_{k} \in \mathcal{U}) \geq 1- \epsilon_{u}, \label{eq:opt_control_constraint_optsteer}
\end{align}
\end{subequations}
where $\Phi_{k}(.)$ is the parameterization for the finite-horizon controller such that $ \mathbf{u}_{k} \coloneqq \Phi_{k}(\mathbf{x}_{k}) \ \forall \ k = 0, 1, \cdots N-1$, and $\mathscr{R}_{\mathcal{G}}$ is a compact, convex subset of the state space $ \mathscr{R}_{\mathcal{G}} \subset \mathcal{X}$, and $\Sigma_{\mathcal{G}} \in \mathbb{S}^{n}_{++}$. For the remainder of this paper, all regions defining steering maneuvers are assumed to have an ellipsoidal representation. We make this explicit for the above program by denoting $\mathscr{R}_{\mathcal{G}}$ as $\mathscr{R}_{\mathcal{G}} \coloneqq \mathscr{R}( \mu_{\mathcal{G},c}, \mathscr{P}_{\mathcal{G}})$, where $ \mu_{\mathcal{G},c} \in \mathbb{R}^{n} $ is the center, and $\mathscr{P}_{\mathcal{G}} \in \mathbb{S}^{n}_{++}$ is the shape matrix of the ellipsoid s.t. $ \mathscr{R}( \mu_{\mathcal{G},c}, \mathscr{P}_{\mathcal{G}}) \coloneqq \{ x \in \mathbb{R}^{n} \ \vert \ (x - \mu_{\mathcal{G},c})\t \mathscr{P}\inv_{\mathcal{G}} (x - \mu_{\mathcal{G},c}) \leq 1\} $. 

Note that the terminal goal reaching constraint $\mu_{N} \in \mathscr{R}( \mu_{\mathcal{G},c}, \mathscr{P}_{\mathcal{G}})$ in the above program is a slightly modified version of the $\operatorname{OPTSTEER}$ defined in \cite{aggarwal2024sdp} requiring exact terminal mean reaching $\mu_{N} = \mu_{\mathcal{G},c}$. This formulation provides increased flexibility by allowing to construct a backward reachable tree of ambiguity sets of distributions as we will see in later sections. Intuitively speaking, the roadmap now looks like a backward reachable tree of ellipsoidal regions representing the ambiguity in the first-order moment such that the tree has favourable space-filling properties as demonstrated experimentally. From a practical standpoint, both the formulations are equivalent in the sense that if the terminal mean is required to reach the prescribed goal mean exactly, this requirement can be approximately encoded by a goal ellipsoid of arbitrarily small volume centered at the prescribed goal mean. From a computational standpoint, the above program remains tractable to solve such that the convex relaxation provided in \cite{rapakoulias2023discrete} would return a feasible control sequence for the above defined program since $ \mu_{N} \in \mathscr{R} ( \mu_{\mathcal{G},c}, \mathscr{P}_{\mathcal{G}}) $ is a convex constraint.

$ \operatorname{OPTSTEER}(q, \mathcal{G}, N) $ has a solution if and only if the constraint set of the optimization problem is non-empty in the space of the decision variables $\{\Phi_{k}\}_{k=0}^{N-1}$. 
However, with control constraints, the existence of an $N$-step controller that steers the system from an initial distribution to a goal needs to be established explicitly by solving a feasibility question defined by  (\ref{eq:opt_dynamics_optsteer})--(\ref{eq:opt_control_constraint_optsteer}). This paper presents solutions to the following problem:
\begin{problem}\label{problem:problem_eng}
    Find feasible paths to the (ambiguity set corresponding to the) goal distribution $\mathcal{G}$ from all query initial distributions $q$ for which paths (respecting all the constraints) exist.
\end{problem}
Note that we use the terms \textit{path} and \textit{control sequence} interchangeably in this paper. Problem~\ref{problem:problem_eng} is an instance of multi-query planning, and it is difficult to solve in general.
We employ graph-based methods to build a roadmap in the space of distributions of the state of the system such that the roadmap could be used across query initial distributions.
The existence of an $N$-step steering control that drives the state distribution from $q$ to $\mathcal{G}$ satisfying all the constraints is determined by the feasibility of (\ref{eq:opt_dynamics_optsteer})--(\ref{eq:opt_control_constraint_optsteer}) that represent a non-convex set in the space of the decision variables. Determining if a non-convex set is empty or not from its algebraic description is an NP-hard problem in general, and the complexity of this feasibility check (finding a feasible path) scales up rapidly with the time-horizon $N$ since the problem size becomes larger and the number of variables increase. 

Problem~\ref{problem:problem_eng} is concerned with finding feasible paths from all possible initial distributions, and our solution methodology proceeds by building a backward reachable tree of feasible paths from the goal distribution and sequencing them together to find a feasible path at run-time for the query initial distribution. This approach avoids the search for the feasible path of the query distribution from scratch (see Section~\ref{sec:approach} for details). 

The paper is organized as follows: In Section~\ref{sec:prelim}, we develop a novel nonlinear program for a distributionally robust covariance steering controller and provide a convex relaxation thereof. Section~\ref{sec:approach} details the construction of a backward reachable tree of ambiguity sets based on novel optimization programs for edge construction that are based on the nonlinear program developed in Section~\ref{sec:prelim}. In Section~\ref{sec:proof}, we give a proof of the Maximum Coverage Theorem stated in \cite{aggarwal2024sdp_CDC} and Section~\ref{sec:experiments} discusses experimenal results.

\section{Novel Distributionally Robust Covariance Steering Controller}\label{sec:prelim}
To solve Problem~\ref{problem:problem_eng}, we build a backward reachable tree of \textit{ambiguity sets of distributions} from the goal, such that each directed edge represents an edge controller that is distributionally robust to the pair of ambiguity sets stored at the corresponding nodes. In this section, we formulate a novel nonlinear program corresponding to the distributionally robust covariance steering maneuver between and \textit{initial} and a \textit{goal} ambiguity set. The distributionally robust program introduced in this section is further developed in Section~\ref{sec:approach} into a procedure to construct a backward reachable tree from the goal by adding maximal size ambiguity sets as nodes and the corresponding maximally distributionally robust edge controllers.

We recall that the ambiguity set of distributions for our formulation is defined completely in terms of ambiguity in the first and the second-order moments respectively. For our approach, the ambiguity in the first-order moment is represented by a compact, convex subset of the state space, and we assume no ambiguity in the second-order moment.
\begin{equation}
    \min _{ \Phi_{k} } \quad J=\sup_{ \tilde{\mathbb{P}} \in \mathcal{P}^{\mathbf{x}_{0:N}, \mathbf{u}_{0:N-1}}} \mathbb{E}_{\tilde{\mathbb{P}}}\left[\sum_{k=0}^{N-1} \mathbf{x}_k^{\top} Q_k \mathbf{x}_k + \mathbf{u}_k^{\top} R_k \mathbf{u}_k\right] \tag{EDGESTEER} \label{eq:J_dro}
\end{equation}
such that, for all $k = 0, 1, \ldots, N - 1$,
\begin{subequations}
\begin{align}
    &\mathbf{x}_{k+1} = A\mathbf{x}_{k} + B\mathbf{u}_{k} + D \mathbf{w}_{k}, \label{eq:opt_dynamics} \\
    &\mathbf{x}_{0} \sim \mathbb{P}_{\mathbf{x}_{0}} \text{ s.t. } \mathbb{P}_{\mathbf{x}_{0}} \in \mathcal{P}^{\mathbf{x}_{0}}, \\
    &\mathbb{P}_{\mathbf{x}_{N}} \in \mathcal{P}^{\mathcal{G}}, \\
    &\inf_{ \mathcal{P}^{\mathbf{x}_{k}} } \mathbb{P}(\mathbf{x}_{k} \in \mathcal{X}) \geq 1- \epsilon_{x}, \label{eq:opt_state_constraint} \\
    &\inf_{ \mathcal{P}^{\mathbf{u}_{k}} } \mathbb{P}(\mathbf{u}_{k} \in \mathcal{U}) \geq 1- \epsilon_{u}, \label{eq:opt_control_constraint}
\end{align}
\end{subequations}
where $\mathcal{P}^{\mathbf{x}_{0:N}, \mathbf{u}_{0:N-1}}$ is the ambiguity set corresponding to the joint distribution of $\{\mathbf{x}_{k}\}_{k=0}^{N}$ and $\{\mathbf{u}_{k}\}_{k=0}^{N-1}$, $ \mathcal{P}^{\mathbf{x}_{0}} = \{ \mathcal{N}_{n}(\mu_{0}, \Sigma_{0}) \ \vert \ \mu_{0} \in \mathscr{R}( \mu_{q}, \mathscr{P}_{q}), \Sigma_{0} = \Sigma_{q} \} $ and $ \mathcal{P}^{\mathcal{G}} = \{ \mathcal{N}_{n}(\mu, \Sigma) \ \vert \ \mu \in \mathscr{R}( \mu_{\mathcal{G},c}, \mathscr{P}_{\mathcal{G}}), \lambda_{\mathrm{max}}(\Sigma) \leq \lambda_{\mathrm{min}}(\Sigma_{\mathcal{G}}) \} $.

$\mathcal{P}^{\mathbf{x}_{k}}$ and $\mathcal{P}^{\mathbf{u}_{k}}$ represent the ambiguity sets corresponding to the state and feedback control distribution at time $k$ respectively, and can be obtained from the ambiguity set of the initial state distribution $\mathcal{P}^{\mathbf{x}_{0}}$ by propagating the ellipsoid $\mathscr{R}(\mu_{q}, \mathscr{P}_{q})$ representing the ambiguity in the first-order moment and the covariance $\Sigma_{q}$ through the dynamics.
\subsection{The Steering Nonlinear Program and Feasibility}\label{sec:cs_review}
We consider an affine state feedback parameterization for the controller to solve \ref{eq:J_dro} as follows, 
\begin{equation}\label{eq:control_param}
    \mathbf{u}_{k} = K_{k}(\mathbf{x}_{k} - \mu_{k}) + v_{k}
\end{equation}
where $ K_{k} \in \mathbb{R}^{n \times m} $ is the feedback gain that controls the covariance dynamics, and $ v_{k} \in \mathbb{R}^{n}$ is the feedforward gain that controls the mean dynamics. 

Under a control law of the form (\ref{eq:control_param}), the state distribution remains Gaussian at all times. Moreover, the ambiguity set of the state distribution at time $k$ can be represented as $\mathcal{P}^{\mathbf{x}_{k}} \coloneqq \{ \mathcal{N}_{n}(\mu_{k}, \Sigma) \ \vert \ \mu_{k} \in \mathscr{R}(\mu_{k,c}, \mathscr{P}_{k}), \Sigma = \Sigma_{k}\}$, and can be obtained by propagating $\mathcal{P}^{\mathbf{x}_{0}} = \{ \mathcal{N}_{n}(\mu_{0}, \Sigma_{0}) \ \vert \ \mu_{0} \in \mathcal{R}(\mu_{q}, \mathscr{P}_{q}), \Sigma_{0} = \Sigma_{q} \} $ through the dynamics as follows \cite{kurzhanski1996ellipsoidal},
\begin{align}
    &\mu_{k} \in \mathscr{R}(\mu_{k,c}, \mathscr{P}_{k}), \\
    &\mu_{k+1,c} = A\mu_{k,c} + Bv_{k}, \ 
    \mathscr{P}_{k+1} = A \mathscr{P}_{k} A\t, \\
    &\mu_{0,c} = \mu_{q}, \mathscr{P}_{0} = \mathscr{P}_{q}, \\
    &\begin{aligned}
        \Sigma_{k+1} = A &\Sigma_{k} A^{\intercal} + B K_{k}\Sigma_{k} A\t + A\Sigma_{k} K\t_{k} B\t \nonumber \\ &+ B K_{k} \Sigma_{k} K\t_{k} B\t + DD\t,
    \end{aligned} \\
    &\Sigma_{0} = \Sigma_{q}.
\end{align}
Moreover, we can simplify the expected cost term within the supremum in \ref{eq:J_dro} completely in terms of the first and second order moments of the state process as 
$\sum_{k=0}^{N-1} \operatorname{tr}({Q_{k}\Sigma_{k}}) + \operatorname{tr}(R_{k} K_{k} \Sigma_{k} K^{\intercal}_{k} ) + \mu^{\intercal}_{k} Q_{k} \mu_{k} + v_{k}^{\intercal} R_{k} v_{k}
$.

We consider polytopic state and control constraints of the form $ \mathcal{X} \coloneqq \{ x_{k} \in \mathbb{R}^{n} \ \vert \ \alpha^{\intercal}_{x} x_{k} \leq \beta_{x} \} $, $ \mathcal{U} \coloneqq \{ u_{k} \in \mathbb{R}^{m} \ \vert \ \alpha^{\intercal}_{u} u_{k} \leq \beta_{u} \} $ such that,
\begin{align}
    \mathbb{P}( \alpha^{\intercal}_{x} \mathbf{x}_{k} \leq \beta_{x} ) &\geq 1 - \epsilon_{x} \label{eq:state_chance_constraint}, \\
    \mathbb{P}( \alpha^{\intercal}_{u} \mathbf{u}_{k} \leq \beta_{u} ) &\geq 1 - \epsilon_{u}, \label{eq:control_chance_constraint}
\end{align}
where $ \alpha_{x} \in \mathbb{R}^{n} $, $ \alpha_{u} \in \mathbb{R}^{m} $, and $\beta_{x}, \beta_{u} \in \mathbb{R}$. $\epsilon_{x}, \epsilon_{u} \in \left[0, 0.5\right]$ represent the tolerance levels with respect to state and control constraint violation respectively, and $ \{\alpha_{x} \mathbf{x}_{k}\}^{N-1}_{k=0} $ and $ \{\alpha_{u} \mathbf{u}_{k}\}^{N-1}_{k=0} $ are univariate random variables with 
first and second order moments,
\begin{align}
\mathbb{E}( \alpha_{x} \mathbf{x}_{k} ) &= \alpha_{x} \mu_{k} \label{eq:state_constraint_first_moment} \\
\mathbb{E}( \alpha_{u} \mathbf{u}_{k} ) &= \alpha_{u} v_{k} \label{eq:control_constraint_first_moment} \\
\mathbb{E}( \alpha^{\intercal}_{x} ( \mathbf{x}_{k} - \mu_{k} ) ( \mathbf{x}_{k} - \mu_{k} )^{\intercal} \alpha_{x}) &= \alpha^{\intercal}_{x} \Sigma_{k} \alpha_{x} \label{eq:state_constraint_second_moment} \\
\mathbb{E}( \alpha^{\intercal}_{u} ( \mathbf{u}_{k} - v_{k} ) ( \mathbf{u}_{k} - v_{k} )^{\intercal} \alpha_{u}) &= \alpha\t_{u} K_{k} \Sigma_{k} K\t_{k} \alpha_{u} \label{eq:control_constraint_second_moment}.
\end{align}
Note that the chance constraints (terms within the infimum) can be written as~\cite{okamoto2018optimal},
\begin{subequations}
\begin{align}
    \Phi^{-1}(1-\epsilon_{x}) \sqrt{ \alpha^{\intercal}_{x} \Sigma_{k} \alpha_{x} } + \alpha^{\intercal}_{x} \mu_{k} - \beta_{x} &\leq 0, \label{eq:state_soft_simplified} \\
    \Phi^{-1}(1-\epsilon_{u}) \sqrt{ \alpha\t_{u} K_{k} \Sigma_{k} K\t_{k} \alpha_{u} } + \alpha^{\intercal}_{u} v_{k} - \beta_{u} &\leq 0, \label{eq:control_soft_simplified}
\end{align}
\end{subequations}
where $ \Phi^{-1}(.)$ is the inverse cumulative distribution function of the normal distribution. Therefore, the optimization problem \ref{eq:J_dro} can be recast as the  nonlinear program,
\begin{align}
    \min_{ \Sigma_{k}, K_{k}, \mu_{k,c}, \mathscr{P}_{k}, v_{k} } J = \sum_{k=0}^{N-1} &\operatorname{tr}({Q_{k}\Sigma_{k}}) + \operatorname{tr}(R_{k} K_{k} \Sigma_{k} K^{\intercal}_{k} ) \nonumber \\ &\hspace*{-1in} + \sup_{\mu_{k} \in \mathscr{R}(\mu_{k,c}, \mathscr{P}_{k})} \mu^{\intercal}_{k} Q_{k} \mu_{k} + v_{k}^{\intercal} R_{k} v_{k}, \label{eq:J_simp}
\end{align}
such that for all $k = 0, 1, \cdots, N - 1$,
\begin{subequations}
\begin{align}
    &\begin{aligned}
        \Sigma_{k+1} = A &\Sigma_{k} A^{\intercal} + B K_{k}\Sigma_{k} A\t + A\Sigma_{k} K\t_{k} B\t \nonumber \\ &+ B K_{k} \Sigma_{k} K\t_{k} B\t + DD\t,
    \end{aligned} \tag{\ref{eq:J_simp}a} \label{eq:covar_prop_nlp} \\
    &\Sigma_{0} = \Sigma_{q}, \tag{\ref{eq:J_simp}b} \label{eq:covar_init_nlp} \\
    &\lambda_{\mathrm{max}}(\Sigma_{N}) \leq \lambda_{\mathrm{min}}(\Sigma_{\mathcal{G}}), \tag{\ref{eq:J_simp}c} \label{eq:covar_goal_nlp} \\
    &\mu_{k} \in \mathscr{R}(\mu_{k,c}, \mathscr{P}_{k}), \tag{\ref{eq:J_simp}d} \label{eq:mean_init_nlp} \\
    &\mu_{k+1,c} = A \mu_{k,c} + B v_{k}, \ \mathscr{P}_{k+1} = A \mathscr{P}_{k} A\t, \tag{\ref{eq:J_simp}e} \label{eq:mean_prop_nlp} \\
    &\mu_{0,c} = \mu_{q}, \ \mathscr{P}_{0} = \mathscr{P}_{q}, \ \tag{\ref{eq:J_simp}f} \\
    &\mathscr{R}(\mu_{N,c}, \mathscr{P}_{N}) \subseteq  \mathscr{R}( \mu_{\mathcal{G},c}, \mathscr{P}_{\mathcal{G}}) \tag{\ref{eq:J_simp}g} \label{eq:ellipsoidal_mean_containment_nlp}, \\
    &\Phi^{-1}(1-\epsilon_{x}) \sqrt{ \alpha^{\intercal}_{x} \Sigma_{k} \alpha_{x} } + \!\!\! \!\!\! \sup_{\mu_{k} \in \mathscr{R}(\mu_{k,c}, \mathscr{P}_{k})} \!\!\! \alpha^{\intercal}_{x} \mu_{k} - \beta_{x} \leq 0, \tag{\ref{eq:J_simp}h} \label{eq:state_constraint_nlp} \\
    &\Phi^{-1}(1-\epsilon_{u}) \sqrt{ \alpha\t_{u} K_{k} \Sigma_{k} K\t_{k} \alpha_{u} } + \alpha^{\intercal}_{u} v_{k} - \beta_{u} \leq 0. \tag{\ref{eq:J_simp}i} \label{eq:control_constraint_nlp}
\end{align}
\end{subequations}
Eq.~(\ref{eq:state_constraint_nlp}) is referred to as the \textit{robust satisfaction of the state chance constraint} and (\ref{eq:ellipsoidal_mean_containment_nlp}) is the \textit{ellipsoidal mean containment constraint}.
%
Eq. (\ref{eq:ellipsoidal_mean_containment_nlp}) can be written as
$ (\mu - \mu_{N,c})\t \mathscr{P}\inv_{N}(\mu - \mu_{N,c}) \leq 1 
\implies (\mu - \mu_{\mathcal{G},c})\t \mathscr{P}\inv_{\mathcal{G}}(\mu - \mu_{\mathcal{G},c}) \leq 1 $, which, using the S-procedure \cite{boyd2004convex, yakubovich_slemma, polik2007survey}, can be rewritten as the 
non-convex matrix inequality 
\begin{equation} \!\!\!
    \begin{bmatrix}
        \mathscr{P}\inv_{\mathscr{G}} & - \mathscr{P}\inv_{\mathscr{G}} \mu_{\mathscr{G},c} \\ 
        -\mu_{\mathscr{G},c}\t \mathscr{P}\inv_{\mathscr{G}} & \!\!\! \mathbb{M}_{1} \! (\mu_{\mathcal{G},c}, \mathscr{P}_{\mathcal{G}})
    \end{bmatrix} \! \preceq \! \lambda \! \begin{bmatrix}
        \mathscr{P}\inv_{N} & - \mathscr{P}\inv_{N} \mu_{N,c} \\ 
        -\mu_{N,c}\t \mathscr{P}\inv_{N} & \!\!\! \mathbb{M}_{1}(\mu_{N,c}, \mathscr{P}_{N})
    \end{bmatrix},
    \label{eq:slemma_matineq}
\end{equation}
where $\mathbb{M}_{1}(\mu, \mathscr{P}) \coloneqq \mu\t \mathscr{P}\inv\mu - 1$,
in the decision variables $\mathscr{P}_{N}$, $\mu_{N,c}$, and $\lambda\geq 0$.
\subsection{Convex Relaxation} \label{sec:convex_relax}
For the nonlinear program (\ref{eq:J_simp}), the covariance term $\Sigma_{k}$ and the control feedback term $K_{k}$ are coupled through the objective function, and the constraints (\ref{eq:covar_prop_nlp}) and (\ref{eq:control_constraint_nlp}). Refs. \cite{liu2022optimal,rapakoulias2023discrete} propose a change of variables $ U_{k} = K_{k} \Sigma_{k} $ and introduce an auxiliary variable $Y_{k}$ to relax (\ref{eq:J_simp}) to the following program where the objective function is convex in $\Sigma_{k}$ and $Y_{k}$ and the constraints (\ref{eq:covar_prop_nlp})--(\ref{eq:control_constraint_nlp}) take the following form,
\begin{align}
    \min_{ \Sigma_{k}, U_{k}, Y_{k}, \mu_{k,c}, \mathscr{P}_{k}, v_{k} } J = & \sum_{k=0}^{N-1} \operatorname{tr}({Q_{k}\Sigma_{k}}) + \operatorname{tr}(R_{k} Y_{k}) \nonumber \\ &+ \sup_{\mu_{k} \in \mathscr{R}(\mu_{k,c}, \mathscr{P}_{k})} \mu^{\intercal}_{k} Q_{k} \mu_{k} + v_{k}^{\intercal} R_{k} v_{k}, \label{eq:J_relaxed}
\end{align}
such that for all $k = 0, 1, \cdots, N - 1$,
\begin{subequations}
\begin{align}
    &C_{k} \triangleq U_{k} \Sigma_{k}\inv U\t_{k} - Y_{k} \preceq 0 \tag{\ref{eq:J_relaxed}a} \label{eq:Ck}, \\
    &G_{k} \triangleq A\Sigma_{k}A\t + B U_{k} A\t + A U\t_{k} B\t \nonumber \\ &~~+ B Y_{k} B\t + D D\t - \Sigma_{k+1} = 0, \tag{\ref{eq:J_relaxed}b} \label{eq:Gk} \\
    &\Phi\inv(1-\epsilon_{u}) \sqrt{ \alpha\t_{u} Y_{k} \alpha_{u} } + \alpha^{\intercal}_{u} v_{k} - \beta_{u} \leq 0, \label{eq:control_soft_relaxed} \tag{\ref{eq:J_relaxed}c} \\
    &\text{subject to (\ref{eq:covar_init_nlp})--(\ref{eq:state_constraint_nlp})}, 
    \notag
\end{align}
\end{subequations}
where the constraint (\ref{eq:Ck}) can be expressed as an linear matrix inequality (LMI) using the Schur complement as follows,
$\left[\begin{array}{ll}\Sigma_k & U_k^{\top} \\ U_k & Y_k\end{array}\right] \succeq 0.$
Due to the presence of the square root, neither  (\ref{eq:state_constraint_nlp}) nor (\ref{eq:control_soft_relaxed}) are convex in $\Sigma_{k}$ and $Y_{k}$ respectively.
Ref.~\cite{rapakoulias2023discrete} proposes a linearization process wherein, 
since the square root is a concave function, the tangent line serves as the global overestimator,
\begin{equation}
    \sqrt{x} \leq \frac{1}{2 \sqrt{x_0}} x+\frac{\sqrt{x_0}}{2}, \quad \forall x, x_0 > 0. \label{eq: sqrt_concave_ub}
\end{equation}
Therefore, constraints (\ref{eq:state_constraint_nlp}) and (\ref{eq:control_soft_relaxed}) are approximated as,
\begin{align}
&\begin{aligned}
\Phi^{-1}\left(1-\epsilon_x\right) & \frac{1}{2 \sqrt{\alpha_x^{\top} \Sigma_r \alpha_x}} \alpha_x^{\top} \Sigma_k \alpha_x+ \sup_{\mu_{k} \in \mathscr{R}(\mu_{k,c}, \mathscr{P}_{k})} \alpha_x^{\top} \mu_k \\
& -\left(\beta_x-\Phi^{-1}\left(1-\epsilon_x\right) \frac{1}{2} \sqrt{\alpha_x^{\top} \Sigma_r \alpha_x}\right) \leq 0,
\end{aligned} \label{eq:sigma_relax}
\end{align}
\begin{align}
&\begin{aligned}
\Phi^{-1}\left(1-\epsilon_u\right) & \frac{1}{2 \sqrt{\alpha_u^{\top} Y_r \alpha_u}} \alpha_u^{\top} Y_k \alpha_u+\alpha_u^{\top} v_k \\
& -\left(\beta_u-\Phi^{-1}\left(1-\epsilon_u\right) \frac{1}{2} \sqrt{\alpha_u^{\top} Y_r \alpha_u}\right) \leq 0,
\end{aligned} \label{eq:Y_relax}
\end{align}
where $\Sigma_{r}, Y_{r}$ are some reference values. The constraints (\ref{eq:sigma_relax}) and (\ref{eq:Y_relax}) are now convex (linear) in $\Sigma_{k}$ and $Y_{k}$, and we represent them as $ \ell\t\Sigma_{k}\ell + \sup_{\mu_{k} \in \mathscr{R}(\mu_{k,c}, \mathscr{P}_{k})} \alpha\t_{x}\mu_{k} - \overline{\beta}_{x} \leq 0 $ and $ e \t Y_{k} e + \alpha\t_{u}v_{k} - \overline{\beta}_{u} \leq 0 $ for future reference.

As a contribution of this paper, we now focus on the convex reformulation of the remainder of the nonlinear program (\ref{eq:J_simp}), specifically the constraints due to the ambiguity in the first order moment of the state distribution (\ref{eq:ellipsoidal_mean_containment_nlp}) and (\ref{eq:sigma_relax}). We provide a convex reformulation of the ellipsoidal mean containment (\ref{eq:ellipsoidal_mean_containment_nlp}) and the robust satisfaction of the state chance constraint (\ref{eq:sigma_relax}) that is sufficient for the original nonconvex constraints to hold.

\noindent
\textbf{Convex Relaxation for Robust Satisfaction of the State Chance Constraint} (\ref{eq:sigma_relax}): The constraint can be written as $\inf_{\mu_{k} \in \mathscr{R}(\mu_{k,c}, \mathscr{P}_{k})}  f(\mu_{k}) \geq - \bar{\beta}_{x} + \ell\t \Sigma_{k} \ell$ where $f(\mu_{k}) \coloneqq - \alpha\t_{x} \mu_{k}$. A sufficient condition for this to hold is for a lower bound on the LHS to be greater than or equal to the RHS. Since $\mathscr{R}(\mu_{k,c}, \mathscr{P}_{k})$ represents an ellipsoidal region with the following parameterization: $ \mathscr{R}(\mu_{k,c}, \mathscr{P}_{k}) \coloneqq \{ \mu \ \vert \ (\mu - \mu_{k,c})\t \mathscr{P}\inv_{k} (\mu - \mu_{k,c}) \leq 1 \}  $, $\inf_{\mu_{k} \in \mathscr{R}(\mu_{k,c}, \mathscr{P}_{k})} f(\mu_{k}) $ is a \textit{quadratically constrained quadratic program} (QCQP). Since the quadratic inequality represents an ellipsoid with a non-empty interior, Slater's condition holds and therefore from strong duality, the infimum of the QCQP is equal to the supremum value of its corresponding dual function. Let $g(\lambda)$ denote the dual function of the QCQP as follows,
    \begin{align*}
        g(\lambda) = \inf_{\mu_{k}} \mathcal{L}(\mu_{k}, \lambda),
    \end{align*}
    for $\lambda \geq 0$ where $\mathcal{L}(\cdot,\cdot)$ is the Lagrangian s.t. $\mathcal{L}(\mu_{k}, \lambda) \!\! \coloneqq -\alpha\t_{x} \mu_{k} + \lambda \left[ (\mu_{k}-\mu_{k,c})\t \mathscr{P}\inv_{k} (\mu_{k}-\mu_{k,c}) - 1 \right] $. From the Lagrangian, we obtain the closed-form expression for the dual of the stated QCQP as, $g(\lambda) = -\frac{\alpha\t_{x}\mathscr{P}_{k} \alpha_{x} }{4\lambda} \! - \! \lambda - \alpha\t_{x}\mu_{k,c}$ for $\lambda > 0$. Since strong duality holds, $ \inf_{\mu_{k} \in \mathscr{R}(\mu_{k,c}, \mathscr{P}_{k})} f(\mu_{k}) = \sup_{\lambda \geq 0} g(\lambda) = - \sqrt{ \alpha\t_{x} \mathscr{P}_{k} \alpha_{x}} - \alpha\t_{x} \mu_{k,c}$. Therefore, the constraint (\ref{eq:sigma_relax}) $\inf_{\mu_{k} \in \mathscr{R}(\mu_{k,c}, \mathscr{P}_{k})}  f(\mu_{k}) \geq - \bar{\beta}_{x} + \ell\t \Sigma_{k} \ell $ can be written as $ \ell\t\Sigma_{k}\ell + \sqrt{ \alpha\t_{x} \mathscr{P}_{k} \alpha_{x}} + \alpha\t_{x} \mu_{k,c} \leq \bar{\beta}_{x} $. This reformulation is nonconvex in decision variable $ \mathscr{P}_{k} $ due to the presence of the square root. From (\ref{eq: sqrt_concave_ub}), the following is a sufficient condition for the non-convex constraint to hold,
    \begin{equation}
        \ell\t\Sigma_{k}\ell + \frac{1}{2\sqrt{ \alpha\t_{x} \mathscr{P}_{r} \alpha_{x}}} \alpha\t_{x} \mathscr{P}_{k} \alpha_{x} + \alpha\t_{x} \mu_{k,c} \leq \bar{\beta}_{x} - \frac{\sqrt{ \alpha\t_{x} \mathscr{P}_{r} \alpha_{x}}}{2},
    \end{equation}
    for some reference value $\mathscr{P}_{r}$. We write the above constraint as $ \ell\t\Sigma_{k}\ell + \ell\t_{\mathscr{R}} \mathscr{P}_{k} \ell_{\mathscr{R}} + \alpha\t_{x} \mu_{k,c} \leq \tilde{\beta}_{x}$ for all future reference. 

\noindent
\textbf{Convex Relaxation for Ellipsoidal Mean Containment} (\ref{eq:ellipsoidal_mean_containment_nlp}):
Note that $\mathscr{R}(\mu_{N,c}, \mathscr{P}_{N}) \subseteq \mathscr{R}(\mu_{N,c}, \gamma \mathscr{P}_{\mathcal{G}}) \subseteq \mathscr{R}(\mu_{\mathcal{G},c}, \mathscr{P}_{\mathcal{G}})$ for some $\gamma > 0$ is a sufficient condition to ensure $\mathscr{R}(\mu_{N,c}, \mathscr{P}_{N}) \subseteq \mathscr{R}(\mu_{\mathcal{G},c}, \mathscr{P}_{\mathcal{G}})$. For the former containment condition to hold, it is sufficient that $ \mathscr{P}_{N} \preceq \gamma \mathscr{P}_{\mathcal{G}}$, and for the latter containment condition $\mathscr{R}(\mu_{N,c}, \gamma P_{\mathcal{G}}) \subseteq \mathscr{R}(\mu_{\mathcal{G},c}, P_{\mathcal{G}})$, the following must hold from the S-procedure,
\begin{align}
    \begin{bmatrix}
        \mathscr{P}\inv_{\mathscr{G}} & - \mathscr{P}\inv_{\mathscr{G}} \mu_{\mathscr{G},c} \\ 
        -\mu_{\mathscr{G},c}\t \mathscr{P}\inv_{\mathscr{G}} & \mathbb{M}^{\mathcal{G}}_{1}
    \end{bmatrix} \preceq \frac{\lambda}{\gamma} \begin{bmatrix}
        \mathscr{P}\inv_{\mathcal{G}} & - \mathscr{P}\inv_{\mathcal{G}} \mu_{N,c} \\ 
        -\mu_{N,c}\t \mathscr{P}\inv_{\mathcal{G}} & \mathbb{M}_{1}(\mu_{N,c}, \mathscr{P}_{G})
    \end{bmatrix}, \label{eq:slemma1}
\end{align}
where $ \mathbb{M}_{1}(\mu, \mathscr{P}) \coloneqq \mu\t \mathscr{P}\inv \mu -1 \ \forall \ \mu \in \mathbb{R}^{n}, \ \forall \ \mathscr{P} \in \mathbb{S}^{n}_{++} $ and $ \mathbb{M}^{\mathcal{G}}_{1} \coloneqq \mathbb{M}_{1}(\mu_{\mathcal{G},c}, \mathscr{P}_{G}) $.
Rearranging and introducing an auxiliary variable $ \tau = \gamma / \lambda > 0$ yields,
\begin{align}
    \begin{bmatrix}
        (\tau - 1) \mathscr{P}\inv_{\mathscr{G}} & - \mathscr{P}\inv_{\mathscr{G}} (\tau \mu_{\mathscr{G},c} - \mu_{N,c}) \\ 
        -(\tau \mu_{\mathscr{G},c} - \mu_{N,c})\t \mathscr{P}\inv_{\mathscr{G}} & \mathbb{M}_{2}(\tau, \mu_{\mathcal{G},c}, \mathscr{P}_{\mathcal{G}}, \mu_{N,c}, \gamma)
    \end{bmatrix} \preceq 0, \label{eq:slemma2}
\end{align}
where $ \mathbb{M}_{2}(\cdot)
\coloneqq \tau(\mu\t_{\mathscr{G},c} \mathscr{P}\inv_{\mathscr{G}} \mu_{\mathscr{G},c} -1) - \mu\t_{N,c} \mathscr{P}\inv_{\mathcal{G}} \mu_{N,c} + \gamma$. Eq. (\ref{eq:slemma2}) still represents a nonconvex matrix inequality in the decision variables $ \tau$, $\gamma$, and $\mu_{N,c}$ due to the presence of $ - \mu\t_{N,c} \mathscr{P}\inv_{\mathcal{G}} \mu_{N,c}$ in $ \mathbb{M}_{2}$. 

We impose an auxiliary constraint $ \tau < 1 $ such that the upper-left block of the above matrix, $ (\tau-1) \mathscr{P}\inv_{\mathcal{G}} \prec 0$ and is invertible. A necessary (and sufficient condition) for (\ref{eq:slemma2}) to hold under the introduced auxiliary constraint is obtained from Proposition~\ref{prop:schur_psd} stated as follows.
\begin{proposition}(\cite{boyd2004convex,horn2012matrix})\label{prop:schur_psd}
    For any symmetric matrix $M$ of the form, $ M = \begin{bmatrix}
        A & B \\ B\t & C
    \end{bmatrix} $, if $A$ is invertible, then $M \preceq 0$ if and only if $ A \prec 0 $ and $ C - B\t A\inv B \preceq 0 $.
\end{proposition}
From Prop.~\ref{prop:schur_psd}, (\ref{eq:slemma2}) holds iff $ \mathbb{M}_{2}(\tau, \mu_{\mathcal{G},c}, \mathscr{P}_{\mathcal{G}}, \mu_{N,c}, \gamma) \preceq 0$. From the definition of $ \mathbb{M}_{2}(\cdot)$, this is equivalent to the constraint $ - \mu\t_{N,c} \mathscr{P}\inv_{\mathcal{G}} \mu_{N,c} \leq -\tau(\mu\t_{\mathscr{G},c} \mathscr{P}\inv_{\mathscr{G}} \mu_{\mathscr{G},c} -1) - \gamma$. We note that $ -\mu\t_{N,c} \mathscr{P}\inv_{\mathcal{G}} \mu_{N,c}$ is a concave function. We linearize $  -\mu\t_{N,c} \mathscr{P}\inv_{\mathcal{G}} \mu_{N,c}  $ around the point $ \mu_{N,c} = \mu_{\mathcal{G},c} $ and use the fact that the tangent line of a concave function at any point provides a universal upper bound to obtain 
$$  -\mu\t_{N,c} \mathscr{P}\inv_{\mathcal{G}} \mu_{N,c} \leq - 2\mu\t_{\mathscr{G},c}\mathscr{P}\inv_{\mathscr{G}}\mu_{N,c} + \mu\t_{\mathcal{G},c} \mathscr{P}\inv_{\mathcal{G}} \mu_{\mathcal{G},c} \ \forall \ \mu_{N,c}.$$ Therefore, 
a sufficient condition to ensure $ \mathbb{M}_{2}(\tau, \mu_{\mathcal{G},c}, \mathscr{P}_{\mathcal{G}}, \mu_{N,c}, \gamma) \preceq 0$ is
$$ - 2\mu\t_{\mathscr{G},c}\mathscr{P}\inv_{\mathscr{G}}\mu_{N,c} + \mu\t_{\mathcal{G},c} \mathscr{P}\inv_{\mathcal{G}} \mu_{\mathcal{G},c} \leq -\tau(\mu\t_{\mathscr{G},c} \mathscr{P}\inv_{\mathscr{G}} \mu_{\mathscr{G},c} -1) - \gamma. $$ 
Thus  
$\mathbb{M}_{3}(\cdot) \coloneqq 
(\tau + 1) \mu\t_{\mathscr{G},c} \mathscr{P}\inv_{\mathscr{G}} \mu_{\mathscr{G},c} - 2\mu\t_{\mathscr{G},c}\mathscr{P}\inv_{\mathscr{G}}\mu_{N,c} - (\tau - \gamma) \leq 0 $ and $ (\tau-1) \mathscr{P}\inv_{\mathcal{G}} \prec 0$ 
are sufficient conditions for (\ref{eq:slemma2}) that are equivalently expressed as the LMI,
\begin{align} \label{eq:slemma_lmi}
    \begin{bmatrix}
    (\tau - 1) \mathscr{P}\inv_{\mathscr{G}} & -\mathscr{P}\inv_{\mathscr{G}} (\tau \mu_{\mathscr{G},c} -\mu_{N,c} ) \\ -(\tau \mu_{\mathscr{G},c} -\mu_{N,c} )\t \mathscr{P}\inv_{\mathscr{G}} 
         & \mathbb{M}_{3}(\tau, \mu_{\mathcal{G},c}, \mathscr{P}_{\mathcal{G}}, \mu_{N,c}, \gamma)
    \end{bmatrix} \preceq 0.
\end{align}
Therefore, the original ellipsoidal mean containment condition (\ref{eq:ellipsoidal_mean_containment_nlp}) has been reduced to (\ref{eq:slemma_lmi}) that represents a LMI in the decision variables $ \tau, \gamma,$ and $\mu_{N,c}$ alongside the constraints $ \gamma, \tau > 0$, and $ \tau < 1$.

The nonlinear program (\ref{eq:J_simp}) is relaxed to the following convex semidefinite program,
\begin{align}
    \min_{ \Sigma_{k}, U_{k}, Y_{k}, \mu_{k,c}, \mathscr{P}_{k}, v_{k} } J_{\mathrm{convex}} \label{eq:J_final}
\end{align}
such that for all $ k = 0, 1, \cdots, N-1$,
\begin{subequations}
\begin{align}
    &C_{k}(\Sigma_{k}, U_{k}, Y_{k}) \preceq 0, \tag{\ref{eq:J_final}a} \label{eq:J_fa} \\
    &G_{k} (\Sigma_{k+1}, \Sigma_{k}, Y_{k}, U_{k}) = 0, \tag{\ref{eq:J_final}b} \label{eq:J_fb} \\
    & \lambda_{\mathrm{max}}(\Sigma_{N}) \leq \lambda_{\mathrm{min}}(\Sigma_{\mathcal{G}}) \tag{\ref{eq:J_final}c} \label{eq: J_fc}, \\ 
    &\mu_{k+1,c} = A \mu_{k,c} + B v_{k}, \ \mathscr{P}_{k+1} = A \mathscr{P}_{k} A\t \tag{\ref{eq:J_final}d} \label{eq:J_fd} \\
    & \mu_{0,c} = \mu_{q},\  \mathscr{P}_{0} = \mathscr{P}_{q} \tag{\ref{eq:J_final}e} \label{eq:J_fe}, \\
    &\begin{bmatrix}
    (\tau - 1) \mathscr{P}\inv_{\mathscr{G}} & -\mathscr{P}\inv_{\mathscr{G}} (\tau \mu_{\mathscr{G},c} -\mu_{N,c} ) \\ -(\tau \mu_{\mathscr{G},c} -\mu_{N,c} )\t \mathscr{P}\inv_{\mathscr{G}} 
         & \mathbb{M}_{3}(\tau, \mu_{\mathcal{G},c}, \mathscr{P}_{\mathcal{G}}, \mu_{N,c}, \gamma)
    \end{bmatrix} \preceq 0, \tag{\ref{eq:J_final}f} \label{eq:J_ff} \\
    & \mathscr{P}_{N} \preceq \gamma \mathscr{P}_{\mathcal{G}}; \ \gamma, \tau > 0 \text{ and } \tau < 1, \tag{\ref{eq:J_final}g} \label{eq:J_fg} \\
    &\ell\t\Sigma_{k}\ell +   \ell\t_{\mathscr{R}} \mathscr{P}_{k} \ell_{\mathscr{R}} +\alpha\t_{x}\mu_{k,c} - \tilde{\beta}_{x} \leq 0, \tag{\ref{eq:J_final}h} \label{eq:J_fh} \\
    &e\t Y_{k} e + \alpha\t_{u}v_{k} - \overline{\beta}_{u} \leq 0, \notag \tag{\ref{eq:J_final}i} \label{eq:J_fi}
\end{align}
\end{subequations}
where $U_{k} = K_{k}\Sigma_{k}$, $C_{k} \triangleq U_{k} \Sigma_{k}\inv U\t_{k} - Y_{k}$, $G_{k} \triangleq A\Sigma_{k}A\t + B U_{k} A\t + A U\t_{k} B\t + B Y_{k} B\t + D D\t - \Sigma_{k+1}$, and $\mathbb{M}_{3}(\tau, \mu_{\mathcal{G},c}, \mathscr{P}_{\mathcal{G}}, \mu_{N,c}, \gamma) \coloneqq (\tau + 1) \mu\t_{\mathscr{G},c} \mathscr{P}\inv_{\mathscr{G}} \mu_{\mathscr{G},c} - 2\mu\t_{\mathscr{G},c}\mathscr{P}\inv_{\mathscr{G}}\mu_{N,c} -(\tau - \gamma)$, and $J_{\mathrm{convex}}$ is any surrogate convex objective obtained from $J$ defined in (\ref{eq:J_relaxed}) by replacing the $ \sup_{\mu_{k} \in \mathscr{R}(\mu_{k,c}, \mathscr{P}_{k})} \mu\t_{k} Q_{k} \mu_{k}$ term, for e.g. with $ \mu\t_{k,c} Q_{k} \mu_{k,c}$ as a surrogate. Note that our problem statement concerns with finding feasible paths to the goal and under any such convex surrogate objective $J_{\mathrm{convex}}$, (\ref{eq:J_final}) recovers a feasible control sequence for the program (\ref{eq:J_simp}) as stated in Theorem~\ref{theorem:relaxation}, the main contribution of this section being the relaxation of the non-convex feasible set of (\ref{eq:J_simp}) to a convex semidefinite program that recovers feasible control sequences for the original problem.
\begin{theorem}\label{theorem:relaxation}
    The solution to the relaxed convex semi-definite program (\ref{eq:J_final}) gives a feasible control sequence for the original nonlinear program (\ref{eq:J_simp}).
\end{theorem}
\begin{proof}
See Appendix~\ref{sec:appendix}.
\end{proof}

\noindent
\subsection{Notation for Feasible Steering Maneuvers} \label{sec:steering_notation}
This subsection introduces the notation that is needed to present our approach to multi-query planning in Section~\ref{sec:approach}. We define an operator $\operatorname{FEASIBLE}$ for the two types of steering maneuvers as discussed below. Both instances of steering maneuvers use similar notation for the operator $\operatorname{FEASIBLE}$ and we assume for all future reference it is clear which notion is being employed from context based on the passed arguments to the operator.

\noindent
\textbf{Steering a distribution to an ambiguity set}: We define an operator $\operatorname{FEASIBLE} \coloneqq \operatorname{FEASIBLE}(q', \mathcal{P}, N)$ for a distribution $q'$ and an ambiguity set $\mathcal{P}$ that returns $\operatorname{TRUE}$ if there exists a feasible control sequence that steers the system from an initial distribution $q'$ to a terminal distribution within the ambiguity set $\mathcal{P}$ in $N$-steps, otherwise returns $\operatorname{FALSE}$. An equivalent mathematical characterization of the operator $\operatorname{FEASIBLE}(q', \mathcal{P}, N)$ is obtained through the feasibility of the system of equations (\ref{eq:opt_dynamics_optsteer})--(\ref{eq:opt_control_constraint_optsteer}) defined for $(\mu_{q}, \Sigma_{q}) = (\mu_{q'}, \Sigma_{q'})$, and $(\mathscr{R}(\mu_{\mathcal{G},c}, \mathscr{P}_{\mathcal{G}}), \Sigma_{\mathcal{G}}) = (\mathscr{R}(\mu_{\mathcal{P},c}, \mathscr{P}_{\mathcal{P}}), \Sigma_{\mathcal{P}})$. 

For an $N$-step control sequence $\mathscr{C}$, the notation $ q' \xrightarrow[N]{\mathscr{C}} \mathcal{P}$ denotes that the system dynamics initialized at $(\mu_{q'}, \Sigma_{q'})$ and driven by the control sequence $\mathscr{C}$ satisfy the state and control chance constraints (\ref{eq:opt_state_constraint_optsteer})--(\ref{eq:opt_control_constraint_optsteer}) and the terminal goal reaching constraint (\ref{eq:opt_goal_mean})--(\ref{eq:opt_goal_covar}) corresponding to the ambiguity set $\mathcal{P}$, i.e. $\Sigma_{\mathcal{G}} = \Sigma_{\mathcal{P}}$ and $ \mathscr{R}_{\mathcal{G}} \coloneqq \mathscr{R}(\mu_{\mathcal{G},c}, \mathscr{P}_{\mathcal{G}}) = \mathscr{R}(\mu_{\mathcal{P},c}, \mathscr{P}_{\mathcal{P}})$.
\newline \noindent
\textbf{Steering between a pair of ambiguity sets}:
We similarly define the operator $\operatorname{FEASIBLE} \coloneqq \operatorname{FEASIBLE}(\mathcal{Q}, \mathcal{P}, N)$ for a pair of ambiguity sets $\mathcal{Q}$ and $\mathcal{P}$ that returns $\operatorname{TRUE}$ if there exists a feasible control sequence that steers the system from any initial distribution in the ambiguity set $\mathcal{Q}$ to a terminal distribution within the ambiguity set $\mathcal{P}$ in $N$-steps, otherwise returns $\operatorname{FALSE}$. An equivalent mathematical characterization of the operator $\operatorname{FEASIBLE}(\mathcal{Q}, \mathcal{P}, N)$ is obtained through the feasibility of (\ref{eq:covar_prop_nlp})--(\ref{eq:control_constraint_nlp}) defined for $ \Sigma_{q} = \Sigma_{\mathcal{Q}}$, $ \mu_{q} = \mu_{\mathcal{Q},c}$, $\mathscr{P}_{q} = \mathscr{P}_{\mathcal{Q}}$, and $ \Sigma_{\mathcal{G}} = \Sigma_{\mathcal{P}}$ and $ \mathscr{R}(\mu_{\mathcal{G},c}, \mathscr{P}_{\mathcal{G}}) = \mathscr{R}(\mu_{\mathcal{P},c}, \mathscr{P}_{\mathcal{P}})$.

Additionally for an $N$-step control sequence $\mathscr{C}$, we use the notation $ \mathcal{Q} \xrightarrow[N]{\mathscr{C}} \mathcal{P} $ to denote that the mean and covariance dynamics initialized at $ ( \mathscr{R}(\mu_{\mathcal{Q},c}, \mathscr{P}_{\mathcal{Q}}), \Sigma_{\mathcal{Q}} ) $ and driven by the $N$-step control sequence $ \mathscr{C} $ satisfy the state and control chance constraints (\ref{eq:state_constraint_nlp})--(\ref{eq:control_constraint_nlp}) and the terminal goal reaching constraints (\ref{eq:covar_goal_nlp}) and (\ref{eq:ellipsoidal_mean_containment_nlp}) corresponding to the ambiguity set $\mathcal{P}$ i.e. $ \Sigma_{\mathcal{G}} = \Sigma_{\mathcal{P}}$ and $\mathscr{R}(\mu_{\mathcal{G},c}, \mathscr{P}_{\mathcal{G}}) = \mathscr{R}(\mu_{\mathcal{P},c}, \mathscr{P}_{\mathcal{P}})$.

\begin{remark}\label{remark:reuse}
(Reuse) Let $\mathscr{C}_{i, \mathcal{G}}$ be a $N$-step control sequence s.t.\ $ (  \mathscr{R}( \mu_{i,c}, \mathscr{P}_{i} ), \Sigma_{i} ) \xrightarrow[N]{\mathscr{C}_{i, \mathcal{G}}} (  \mathscr{R}( \mu_{\mathcal{G},c}, \mathscr{P}_{\mathcal{G}} ) $, then it follows that $ (\mu, \Sigma) \xrightarrow[N]{\mathscr{C}_{i, \mathcal{G}}} (  \mathscr{R}( \mu_{\mathcal{G},c}, \mathscr{P}_{\mathcal{G}} ) $ $\forall \ \mu \in \mathscr{R}( \mu_{i,c}, \mathscr{P}_{i}), \Sigma \preceq \Sigma_{i}$. 

In other words, a distributionally robust control sequence $ \mathscr{C}_{i, \mathcal{G}} $ computed for a pair of ambiguity sets (initial and goal) provides a feasible control sequence for steering any distribution with the mean instantiated from within the ambiguity set of the first-order moment of the initial distribution, and covariance less than or equal to the covariance of the ambiguity set of the initial distribution (in the sense of the Loewner order) to the goal ambiguity set.
\end{remark}
Remark~\ref{remark:reuse} follows directly from an analysis of (\ref{eq:covar_prop_nlp})--(\ref{eq:control_constraint_nlp}) and (\ref{eq:J_fa})--(\ref{eq:J_fi}). It motivates the design of \textit{maximally} distributionally robust controllers that add edges from the goal in a backwards fashion to discover paths from the \textit{largest} possible set of initial distributions. $\operatorname{MAXELLIPSOID}$ and $\operatorname{MAXCOVARELL}$ are two instantiations of \textit{maximally} distributionally robust controllers that allow maximum reuse of computation for the first-order moment and the second-order moment respectively. Both these controller formulations, as discussed in Section~\ref{sec:approach}, given a goal ambiguity set, search for a maximal size ambiguity set of the initial distribution and the corresponding control sequence that achieves the steering maneuver through novel optimization programs.

\subsection{Definition of $h$-$\operatorname{BRS}$} \label{sec:hBRS}
We now define the following mathematical objects that will aid the analysis and discussion of our proposed approach in Section~\ref{sec:approach}.
\begin{definition}\label{def:hbrs_dist}[\textit{$h$-$\operatorname{BRS}$ of a distribution $p$}] The $h$-hop backward reachable set of distributions ($h$-$\operatorname{BRS}$) of a distribution $p$ is defined as the set of all initial distributions for which there exists a feasible control sequence that can steer the system to $p$ in $hN$ timesteps. A mathematical characterization of $h$-$\operatorname{BRS}(p)$ is obtained as follows,
\begin{align}
    h\text{-}\operatorname{BRS}(p) =  \{ (\mu_{q}, \Sigma_{q}) \ \vert \ \operatorname{FEASIBLE}(q, p, hN) = \operatorname{TRUE} \}. \notag
\end{align}
\end{definition}
\noindent
\textbf{Remark}: We note that the operator $\operatorname{FEASIBLE}$ can be defined for a pair of distributions through the feasibility of the same system of equations (\ref{eq:opt_dynamics_optsteer})--(\ref{eq:opt_control_constraint_optsteer}) as done in Section~\ref{sec:steering_notation}, the only difference being the terminal mean constraint (\ref{eq:opt_goal_mean}), which for the case of a pair of distributions requires exact terminal mean reaching instead of set containment. For a pair of distributions $(q, p)$, the goal reaching constraints become $ \mu_{N} = \mu_{p}$ and $ \Sigma_{\mathcal{G}} = \Sigma_{p}$.

We can similarly define the $h$-hop backward reachable set of distributions for an ambiguity set $\mathcal{P}$ as follows,
\begin{definition}\label{def:hbrs_ambdist}[\textit{$h$-$\operatorname{BRS}$ of an ambiguity set $\mathcal{P}$}] The $h$-hop backward reachable set of distributions ($h$-$\operatorname{BRS}$) of an ambiguity set $\mathcal{P}$ is defined as the set of all initial distributions for which there exists a feasible control sequence that can steer the system to $\mathcal{P}$ in $hN$ timesteps. A mathematical characterization of $h$-$\operatorname{BRS}(\mathcal{P})$ is obtained as follows,
\begin{align}
    h\text{-}\operatorname{BRS}(\mathcal{P}) =  \{ (\mu_{q}, \Sigma_{q}) \ \vert \ \operatorname{FEASIBLE}(q, \mathcal{P}, hN) = \operatorname{TRUE} \}. \notag
\end{align}
\end{definition}
We further define the $h$-$\operatorname{BRS}$ of a tree of distributions $\mathcal{T}$ as follows (for more details on the tree data structure, see Section~\ref{sec:approach}),
\begin{definition}\label{def:hbrs_treedist}[\textit{$h$-$\operatorname{BRS}$ of a tree of distributions $\mathcal{T}$}] The $h$-hop backward reachable set of distributions ($h$-$\operatorname{BRS}$) of a tree of distributions $\mathcal{T}$ is defined as the set of all initial distributions for which there exists a feasible control sequence that can steer the system to the distribution stored at the root node of the tree in $hN$ timesteps. A mathematical characterization of $h$-$\operatorname{BRS}(\mathcal{T})$ is obtained as follows,
\begin{align}
        h\text{-}\operatorname{BRS}(\mathcal{T}) = \bigcup_{i \in \nu(\mathcal{T})} (h - d_{i})\text{-}\operatorname{BRS}(\nu_{i}), \notag
    \end{align}
\end{definition}
\noindent
where $\nu(\mathcal{T})$ is the set of all vertices in the tree, and $d_{i} $ is the distance of the $i$-th node from the root node in terms of the number of hops. Note that we can define the $h$-$\operatorname{BRS}$ of a tree of ambiguity sets in a similar fashion as Definition~\ref{def:hbrs_treedist}.
Further note that, while the approach in this paper concerns growing a tree of ambiguity sets of distributions as developed in Section~\ref{sec:approach}, Definitions~\ref{def:hbrs_dist} and~\ref{def:hbrs_treedist} are key to the analysis of the (Maximum Coverage) Theorem~\ref{theorem:coverage} stated in \cite{aggarwal2024sdp} and as proved in Section~\ref{sec:proof}.

\section{$\operatorname{MAXELLIPSOID}$ BACKWARD REACHABLE TREE (BRT) FOR MULTI-QUERY PLANNING}\label{sec:approach}

We solve Problem~\ref{problem:problem_eng} by constructing a backward reachable tree of ambiguity sets of distributions that verifiably reach the goal $\mathcal{G}$ under constraints on the control input. As discussed previously, in the presence of control constraints, existence of a control sequence that steers the system from an initial distribution to a (ambiguity set) goal distribution is established by solving a feasibility problem. The size of this feasibility check / finding a feasible path scales with the time-horizon and the backward reachable tree enables a faster feasibility check on a long time-horizon by finding feasible paths to any existing node on the BRT instead of directly checking feasibility against the goal distribution.

We refer to this idea as \textit{recursive feasibility} since the branches of the tree can be thought of as carrying a certificate of feasibility along its edges from the goal in a backwards fashion s.t. guaranteeing feasibility to any of the children nodes in the tree guarantees feasibility to all the upstream parent nodes and consequently the root node that corresponds to the (ambiguity set) goal distribution.

We introduce a novel objective function $\operatorname{MAXELLIPSOID}$, as discussed in the following subsection, for adding nodes (\textit{maximal} size ambiguity sets) and constructing corresponding edge controllers  
(\textit{maximally} distributionally robust covariance steering controllers).
\subsection{$\operatorname{MAXELLIPSOID}$: Novel Objective for Construction of the Edge Controller} \label{sec:max_covar}
We define a procedure to construct an $N$-step edge controller as follows. The procedure $\operatorname{MAXELLIPSOID}$ takes in as input the center of the candidate ellipsoidal ambiguity set corresponding to the initial mean  $\mu_{q}$, a candidate initial covariance $\Sigma_{q}$, a goal ambiguity set at the end of the $N$-step steering maneuver $ (\mathscr{R}(\mu_{\mathcal{G},c}, \mathscr{P}_{\mathcal{G}}), \Sigma_{\mathcal{G}})$ and computes a maximal ellipsoid shape matrix $ \mathscr{P}_{0}$, henceforth referred to as $\mathscr{P}_{0, \mathrm{max}}$, and the associated control sequence $\mathscr{C} \coloneqq \{ \mathscr{C}_{k} \}_{k=0}^{N-1}$ that achieves the corresponding distributionally robust steering maneuver. The control at time $k$ is a tuple of the feedback and the feedforward term s.t.\ $\mathscr{C}_{k} = ( K_{k}, v_{k} )$.
\begin{align}
    \min_{ \Sigma_{k}, U_{k}, Y_{k}, \mu_{k,c}, \mathscr{P}_{k}, v_{k} } J_{\mathrm{MAXELLIPSOID}} = -\log\det(\mathscr{P}_{0}) \tag{$\operatorname{MAXELLIPSOID}$} \label{eq:J_maxellipsoid}
\end{align}
such that for all $k = 0, 1, \cdots, N - 1$,
\begin{subequations}
\begin{align}
    &\mu_{k+1,c} = A \mu_{k,c} + B v_{k}, \ \mathscr{P}_{k+1} = A \mathscr{P}_{k} A\t, \\
    &\mu_{0,c} = \mu_{q}, \\
    &\begin{bmatrix}
    (\tau - 1) \mathscr{P}\inv_{\mathscr{G}} & -\mathscr{P}\inv_{\mathscr{G}} (\tau \mu_{\mathscr{G},c} -\mu_{N,c} ) \\ -(\tau \mu_{\mathscr{G},c} -\mu_{N,c} )\t \mathscr{P}\inv_{\mathscr{G}} 
         & \mathbb{M}_{3}(\tau, \mu_{\mathcal{G},c}, \mathscr{P}_{\mathcal{G}}, \mu_{N,c}, \gamma)
    \end{bmatrix} \preceq 0,\\
    &\mathscr{P}_{N} \preceq \gamma \mathscr{P}_{\mathcal{G}}; \ \gamma, \tau > 0 \text{ and } \tau < 1,\\
    &\ell\t\Sigma_{k}\ell +   \ell\t_{\mathscr{R}} \mathscr{P}_{k} \ell_{\mathscr{R}} +\alpha\t_{x}\mu_{k,c} - \overline{\beta}_{x} \leq 0, \\
    &e\t Y_{k} e + \alpha\t_{u}v_{k} - \overline{\beta}_{u} \leq 0, \\
    &\text{subject to (\ref{eq:J_fa})--(\ref{eq: J_fc})}, \notag
\end{align}
\end{subequations}
where the $ -\log\det(\mathscr{P})$ operator is proportional to the volume of the ellipsoid with the shape matrix $\mathscr{P}$. $\log\det(\mathscr{P})$ is a concave function of the positive semidefinite matrix variable $\mathscr{P}$, therefore \ref{eq:J_maxellipsoid} is a convex semidefinite program that can be solved by existing solvers efficiently \cite{cvxpy2016}. For all future references, we use the following notation to denote the solution returned by the $\operatorname{MAXELLIPSOID}$ program (shorthand $\operatorname{MAXELL}$), $$\mathscr{P}_{0, \mathrm{max}}, \mathscr{C}_{\mathrm{max}} \longleftarrow \operatorname{MAXELLIPSOID}( \mu_{q}, \Sigma_{q}, \mu_{\mathcal{G},c}, \mathscr{P}_{\mathcal{G}}, \Sigma_{\mathcal{G}}, N).$$
$\operatorname{MAXELLIPSOID}$ searches for a maximal volume backward reachable ellipsoidal ambiguity set $\mathscr{P}_{0, \mathrm{max}}$ for a given initial covariance $\Sigma_{q}$ and the corresponding steering controller $\mathscr{C}_{\mathrm{max}}$ that reaches a goal ambiguity set in $N$-steps under control constraints. The computed edge controller provides a feasible control sequence for a family of distributions at once such that feasible paths to the goal are found for any query mean instantiated from within the computed maximal size ambiguity set.
\noindent
\subsection{Bi-level search for a Maximally Distributionally Robust Controller}\label{sec:bilevel}
We define a new optimization program $\operatorname{MAXCOVARELL}$ that takes in as input an ellipsoidal ambiguity set corresponding to the initial mean $\mathscr{R}(\mu_{0,c}, \mathscr{P}_{0})$, an ellipsoidal ambiguity set for the goal mean $\mathscr{R}(\mu_{\mathcal{G},c}, \mathscr{P}_{\mathcal{G}})$, and a goal covariance $\Sigma_{\mathcal{G}}$, and returns the initial covariance $\Sigma_{0}$ in a \textit{maximal} sense, henceforth referred to as $\Sigma_{0, \mathrm{max}}$, and an $N$-step control sequence through the following optimization program,
\begin{align}
    \min_{ \Sigma_{k}, U_{k}, Y_{k}, \mu_{k,c}, \mathscr{P}_{k}, v_{k} } J_{\mathrm{MAXCOVARELL}} = -\lambda_{\mathrm{min}}(\Sigma_{0}) \tag{$\operatorname{MAXCOVARELL}$} \label{eq:J_maxcovarell}
\end{align}
such that for all $k = 0, 1, \cdots, N - 1$,
\begin{subequations}
\begin{align}
    &C_{k}(\Sigma_{k}, U_{k}, Y_{k}) \preceq 0, \\
    &G_{k} (\Sigma_{k+1}, \Sigma_{k}, Y_{k}, U_{k}) = 0, \\
    & \lambda_{\mathrm{max}}(\Sigma_{N}) \leq \lambda_{\mathrm{min}}(\Sigma_{\mathcal{G}}), \\
    &\ell\t\Sigma_{k}\ell +   \ell\t_{\mathscr{R}} \mathscr{P}_{k} \ell_{\mathscr{R}} +\alpha\t_{x}\mu_{k,c} - \overline{\beta}_{x} \leq 0, \\
    &e\t Y_{k} e + \alpha\t_{u}v_{k} - \overline{\beta}_{u} \leq 0, \\
    &\text{subject to (\ref{eq:J_fd})--(\ref{eq:J_fg}).}
\end{align}
\end{subequations}

\noindent
$\operatorname{MAXCOVARELL}$ is based on the maximization of the minimum eigenvalue of $\Sigma_{0}$ for a given ellipsoidal ambiguity set for the initial mean, and a goal ambiguity set. It is based on Remark~\ref{remark:reuse} and the observation that if $ \mathcal{Q} \xrightarrow[N]{\mathscr{C}} \mathcal{G} $, then $ \mathcal{Q}' \xrightarrow[N]{\mathscr{C}} \mathcal{G} \ \ \forall \ \ \mathcal{Q}' \text{ s.t. } \mathscr{R}_{\mathcal{Q}'} \subset \mathscr{R}_{\mathcal{Q}} \text{ and } \Sigma_{\mathcal{Q}'} \prec \Sigma_{\mathcal{Q}} $. Therefore for a given ellipsoidal ambiguity set for the initial mean, we aim to find a maximal covariance and an associated control sequence  $ (\Sigma_{0,\mathrm{max}}, \mathscr{C})$ such that the computed $\mathscr{C}$ could be reused across largest family of initial distributions, i.e., find $ \Sigma_{0} $ such that the set $ \{ \Sigma_{0}' \ \vert \ \Sigma_{0}' \prec \Sigma_{0}  \} $ is as large as possible. This leads to the maximization of the minimum eigenvalue of the initial covariance as a natural objective function for our search.

$\operatorname{MAXCOVARELL}$ provides a certificate of reachability for any goal ambiguity set in terms of the maximum permissible value of the minimum eigenvalue of the covariance at any query ellipsoidal ambiguity set for the initial mean for which there exists a feasible control sequence that can achieve the corresponding steering maneuver under control constraints. For instance, consider a goal $ (\mathscr{R}(\mu_{\mathcal{G},c}, \mathscr{P}_{\mathcal{G}}), \Sigma_{\mathcal{G}}) $ and a query ellipsoidal set for the initial mean $\mathscr{R}(\mu_{q}, \mathscr{P}_{q})$, and let $ \Sigma_{0,\mathrm{max}}$, $ \mathscr{C}_{0, \mathrm{max}} $ be such that,
\begin{align*}
    \Sigma_{0, \mathrm{max}}, \mathscr{C}_{0, \mathrm{max}} \longleftarrow \operatorname{MAXCOVARELL}( \mu_{q}, \mathscr{P}_{q},  \mu_{\mathcal{G},c}, \mathscr{P}_{\mathcal{G}} , N).
\end{align*}
It follows from the above that $\forall \ \Sigma \succ 0$ s.t.\ $ \lambda_{\mathrm{min}}(\Sigma) > \lambda_{\mathrm{min}}(\Sigma_{0,\mathrm{max}})$, $ \nexists \ \mathscr{C}$ s.t.\ $ (\mathscr{R}(\mu_{q}, \mathscr{P}_{q}), \Sigma) \xrightarrow[N]{\mathscr{C}} \mathcal{G}$ by definition of $\operatorname{MAXCOVARELL}$ otherwise $ \lambda_{\mathrm{min}} (\Sigma_{0, \mathrm{max}}) $ is not the maximum possible minimum eigenvalue of the initial covariance for the existence of a feasible path and we arrive at a contradiction.
\newline \noindent \textbf{Remark} (\textit{Bi-level search for a maximally distributionally robust controller}) Note that $\operatorname{MAXELLIPSOID}$ and $\operatorname{MAXCOVARELL}$ construct maximal size ambiguity sets and corresponding edge controllers in the first-order and the second-order moment sense respectively. To facilitate the search for a maximally distributionally robust controller, we could follow a bi-level search protocol wherein starting from a given initial covariance, a $\operatorname{MAXELLIPSOID}$ search for a maximal ellipsoid is followed up with a $ \operatorname{MAXCOVARELL} $ procedure to search for a maximal initial covariance corresponding to the discovered maximal ellipsoid.
\subsection{Construction of the MAXELLIPSOID BRT}\label{sec:construction_brt}
The algorithm proceeds by building a backwards reachable tree $\mathcal{T(\mathcal{G})}$ from the goal ambiguity set $\mathcal{G}$ by populating a set of nodes $\{\mathcal{\nu}_{i}\}$, and a set of edge controllers $\{\varepsilon_{i,j}\}$. Each node $i$, $ \nu_{i} $, is a tuple $ \left( \mu_{i,c}, \mathscr{P}_{i}, \Sigma_{i}, p_{i}, \mathscr{C}_{i}, \mathrm{ch}_{i} \right) $ where $ (\mu_{i,c}, \mathscr{P}_{i})$ correspond to the ambiguity set, and $ \Sigma_{i}$ is the covariance stored at the node, $p_{i}$ is a pointer to the parent node, $ \mathscr{C}_{i} \coloneqq \{ K^{i, p_{i}}_{t}, v^{i, p_{i}}_{t} \}_{t=0}^{N-1} $ is the $N$-step control sequence stored at the node that steers the state distribution from $ (\mathscr{R}(\mu_{i,c}, \mathscr{P}_{i}), \Sigma_{i}) $ to $ (\mathscr{R}(\mu_{p_{i},c}, \mathscr{P}_{p_{i}}), \Sigma_{p_{i}}) $, and $ \mathrm{ch}_{i} $ is the list of pointers of all the children node of node $i$ in the tree $\mathcal{T}(\mathcal{G})$.

The variable $ \varepsilon $ is another data structure that stores all the edge information for the tree $\mathcal{T}(\mathcal{G})$, such that, $ \varepsilon_{i,j} \coloneqq \{ K^{i,j}_{t}, v^{i,j}_{t} \}_{t=0}^{N-1} $ stores the $N$-step control sequence that steers the system from node $i$ to node $j$ if such an edge exists, and is empty otherwise.
\begin{algorithm}[t]
\caption{Constructing the MAXELLIPSOID BRT}
\label{alg:construct_brt}
\begin{algorithmic}[1]
    \Require $\mathcal{G}$, $N$, $n_\mathrm{iter}$ 
    \Ensure  $\mathcal{T}$ 
    \State $\nu \gets \phi$, $\varepsilon \gets \phi$
    \State $ \nu_{0} \gets \operatorname{CREATENODE}(\mu_{\mathcal{G},c}, \mathscr{P}_{\mathcal{G}}, \Sigma_{\mathcal{G}}, \operatorname{NONE}, \operatorname{NONE}, \{\})$
    \State $\nu \gets \nu \cup \{\nu_{0}\}$
    \For{$i \gets 1 \textrm{ to } n_\mathrm{iter}$}
        \State 
        $\nu_{k} \gets \operatorname{RAND}(\nu)$ 
        \State $ \mu_{q} \gets \operatorname{RANDMEANAROUND}(\nu_{k}, r_{\textrm{sample}}) $
        \State $ \textrm{status}, \mathscr{P}_{\mathrm{max}}, \mathscr{C} \gets \operatorname{MAXELL}(\mu_{q}, \Sigma_{q}, \mu^{(k)}, \mathscr{P}_{k}, \Sigma^{(k)}, N)$
        \If{$\textrm{status} \neq \textrm{infeasible}$}
            \State $ \textrm{idx} \gets \operatorname{size}(\mathcal{\nu}) + 1 $
            \State $ \mathrm{ch}_{\mathrm{idx}} \gets \{\}$
            \State $ \nu_{\mathrm{new}} \gets \operatorname{CREATENODE}(\mu_{q}, \mathscr{P}_{\mathrm{max}}, \Sigma_{q}, \textrm{idx}, k, \mathrm{ch}_{\mathrm{idx}})$
            \State $ \nu \gets \nu \bigcup \{ \nu_{\mathrm{new}} \} $
            \State $ \mathrm{ch}_{k} \gets \mathrm{ch}_{k} \bigcup \{ \textrm{idx}\} $
            \State $ \varepsilon_{\textrm{idx},k} \gets \mathscr{C} $
            \State $ \varepsilon \gets \varepsilon \bigcup \{ \varepsilon_{\textrm{idx},k} \} $
        \EndIf
    \EndFor
    \State $  \mathcal{T} \gets \{ \nu, \varepsilon \} $
    \State $ \textrm{\textbf{Return}} \ \  \mathcal{T}$
\end{algorithmic}
\end{algorithm}
The following discusses the essential subroutines of this algorithm.
\subsubsection{Node selection}
The tree is grown in the spirit of finding paths from all query initial distributions for which paths would exist to the goal ambiguity set. In our implementation, the nodes are selected randomly according to the Voronoi bias (of the first order moment of the nodes) to bias population of the $1$-$\operatorname{BRS}$ of the node ambiguity sets (see Definition~\ref{def:hbrs_ambdist}) whose corresponding Voronoi regions are relatively unexplored in the sense of the first order moment.

\subsubsection{Node expansion}
Once a node to expand has been selected, a query mean is sampled from a neighbourhood of some radius around it and a connection is attempted through the $\operatorname{MAXELLIPSOID}$ method for edge construction. Let's say the $k$th node on the tree containing the ambiguity set $ ( \mathscr{R}(\mu^{(k)}, \mathscr{P}_{k}), \Sigma^{(k)})$ has been selected to expand, and let $ \mu_{q} $ be the query mean sampled from a neighbourhood around $ \mu^{(k)} $ through the $\operatorname{RANDMEANAROUND}(., r)$ module where $r$ is some sampling radius. We solve the following optimization to construct the edge,
\begin{align}
    \mathscr{P}_{\mathrm{max}}, \mathscr{C}_{\mathrm{max}} \longleftarrow \operatorname{MAXELLIPSOID}( \mu_{q}, \Sigma_{q}, \mu^{(k)}, \mathscr{P}_{k}, \Sigma^{(k)}, N)
\end{align}
where $ \Sigma_{q}$ is an initial covariance chosen by some heuristic (fixed for all candidate nodes in our implementation). 
$(\mathscr{R}(\mu_{q}, \mathscr{P}_{\mathrm{max}}), \Sigma_{q})$ is added as a node to the tree with the edge controller $\mathscr{C}_{\mathrm{max}}$ and $(\mathscr{R}(\mu^{(k)}, \mathscr{P}_{k}), \Sigma^{k})$ as the parent if the status of the above optimization problem (as returned by the solver) is not infeasible. Note that the above edge construction procedure can be \textit{augmented} by an additional step of optimization as a bi-level search for the maximally distributionally robust controller through $\operatorname{MAXCOVARELL}$ as described in Section~\ref{sec:bilevel}.

\noindent
\textbf{Concatenation of control sequences}: A concatenation $\mathscr{C}_{A,B}$ of two control sequences $ \mathscr{C}_{A} $, $\mathscr{C}_{B}$ of lengths $N_{A}, N_{B}$ respectively is a control sequence of length $N_{A} + N_{B}$ represented through the $\bigcup$ operator i.e.\ $\mathscr{C}_{A,B} \coloneqq \mathscr{C}_{A} \bigcup \mathscr{C}_{B}$, s.t.\ $\mathscr{C}_{A,B}(t) \coloneqq C_{A}(t) \ \forall \ t = 0, 1, \cdots N_{A} - 1$, and $\mathscr{C}_{A,B}(t) \coloneqq C_{B}(t-N_{A}) \ \forall \ t = N_{A}, N_{A} + 1, \cdots N_{A} + N_{B} - 1$. Note that the concatenation operator $\bigcup$ is \textit{non-commutative} in the two argument control sequences, i.e.\ $ \mathscr{C}_{A} \bigcup \mathscr{C}_{B} \neq \mathscr{C}_{B} \bigcup \mathscr{C}_{A} $.

To establish / guarantee feasibility of the query distribution to the goal ambiguity set, it is sufficient to guarantee feasibility to any node ambiguity set that is already verified to reach the goal. This is formalized as the following lemma on sequential composition of control sequences ensuring satisfaction of state and control chance constraints along the overall concatenated trajectory (see Fig.~\ref{fig:intro_picture}).
\begin{lemma}\label{lemma:sequential_composition}(Recursive Feasibility through Sequential Composition) Consider the distribution $ (\mu_{i}, \Sigma_{i})$, the ambiguity sets $(\mathscr{R}(\mu_{j,c}, \mathscr{P}_{j}), \Sigma_{j})$ and $(\mathscr{R}(\mu_{k,c}, \mathscr{P}_{k}), \Sigma_{k})$, control sequences $\mathscr{C}_{i,j}, \mathscr{C}_{j,k}$, and $h_{i,j}, h_{j,k}$ such that \ $ (\mu_{i}, \Sigma_{i}) \xrightarrow[h_{i,j} N]{\mathscr{C}_{i,j}} (\mathscr{R}(\mu_{j,c}, \mathscr{P}_{j}), \Sigma_{j}) $, and $ (\mathscr{R}(\mu_{j,c}, \mathscr{P}_{j}), \Sigma_{j}) \xrightarrow[h_{j, k} N]{\mathscr{C}_{j, k}} (\mathscr{R}(\mu_{k,c}, \mathscr{P}_{k}), \Sigma_{k}) $. It follows that $ (\mu_{i}, \Sigma_{i}) \xrightarrow[(h_{i,j} + h_{j, k}) N]{\mathscr{C}_{i,k}} (\mathscr{R}(\mu_{k,c}, \mathscr{P}_{k}), \Sigma_{k})$ where $ \mathscr{C}_{i,k} \coloneqq \mathscr{C}_{i,j} \bigcup \mathscr{C}_{j,k} $.
\end{lemma}
\begin{proof}
    We want to prove that the system initialized at the distribution $( \mu_{i}, \Sigma_{i} )$ and driven by the control sequence $ \mathscr{C}_{i,k} $ which is obtained from the concatenation of the two control sequences $ \mathscr{C}_{i,j} $, and $ \mathscr{C}_{j,k}$ i.e.\ $ \mathscr{C}_{i,k} = \mathscr{C}_{i,j} \bigcup \mathscr{C}_{j,k} $ satisfies all the state and control chance constraints of the form (\ref{eq:opt_state_constraint_optsteer})--(\ref{eq:opt_control_constraint_optsteer}) for a trajectory of length $ (h_{i,j} + h_{j,k}) N $, and the terminal goal reaching constraints of the form (\ref{eq:opt_goal_mean}) and (\ref{eq:opt_goal_covar}) for the goal ambiguity set $(\mathscr{R}(\mu_{k,c}, \mathscr{P}_{k}), \Sigma_{k})$.
    
    The mean and covariance dynamics at any time $t$ as a function of the intial distribution and the control sequence are represented as $ \mu_{t}(\mu_{i}, \mathscr{C}_{i,k})$ and  $ \Sigma_{t}(\Sigma_{i}, \mathscr{C}_{i,k})$ respectively.
    
    For $t = 0, 1, \cdots h_{i,j} N$, $ \mu_{t}(\mu_{i}, \mathscr{C}_{i,k}) = \mu_{t}(\mu_{i}, \mathscr{C}_{i,j}) $, and $ \Sigma_{t}(\Sigma_{i}, \mathscr{C}_{i,k}) = \Sigma_{t}(\Sigma_{i}, \mathscr{C}_{i,j}) $. Since $ \mathscr{C}_{i,j} $ is a feasible control sequence for the $(i,j)$ maneuver s.t.\ $ (\mu_{i}, \Sigma_{i}) \xrightarrow[h_{i, j} N]{\mathscr{C}_{i,j}} (\mathscr{R}(\mu_{j,c}, \mathscr{P}_{j}), \Sigma_{j}) $, all the state and control chance constraints are satisfied by $ \mu_{t}( \mu_{i}, \mathscr{C}_{i,j}) $ and $\Sigma_{t}( \Sigma_{i}, \mathscr{C}_{i,j}) $ for $t = 0, 1, \cdots h_{i, j} N$. Also from the goal reaching constraints, $\mu_{h_{i, j} N} \coloneqq \mu_{h_{i, j} N} (\mu_{i}, \mathscr{C}_{i,j}) \in \mathscr{R}(\mu_{j,c}, \mathscr{P}_{j})$, and $ \Sigma_{h_{i, j} N} \coloneqq \Sigma_{h_{i, j} N} (\Sigma_{i}, \mathscr{C}_{i,j})$ is s.t.\ $ \lambda_{\mathrm{max}}( \Sigma_{h_{i, j} N} ) \leq \lambda_{\mathrm{min}}( \Sigma_{j} )$ which implies $ \Sigma_{h_{i, j} N} \preceq \Sigma_{j} $.

    The rest of the maneuver for $t = h_{i,j }N + 1, \cdots (h_{i,j} + h_{j,k}) N$ could be thought of as an $N_{j,k}$ step maneuver initialized at $ \mu_{h_{i,j}N} $ and $ \Sigma_{h_{i,j}N} $. $ \mathscr{C}_{j,k} $ is s.t.\ $ (\mathscr{R}(\mu_{j,c}, \mathscr{P}_{j}), \Sigma_{j}) \xrightarrow[h_{j, k} N]{\mathscr{C}_{j, k}} (\mathscr{R}(\mu_{k,c}, \mathscr{P}_{k}), \Sigma_{k}) $, therefore from Remark~\ref{remark:reuse}, $ (\mu_{h_{i,j}N}, \Sigma_{h_{i, j} N} ) \xrightarrow[N_{j, k}]{\mathscr{C}_{j, k}} (\mathscr{R}(\mu_{k,c}, \mathscr{P}_{k}), \Sigma_{k}) $ since $\mu_{h_{i,j}N} \in \mathscr{R}(\mu_{j,c}, \mathscr{P}_{j})$ and $\Sigma_{h_{i,j}N} \preceq \Sigma_{j}$ and the lemma follows.
\end{proof}
\subsection{Planning through the BRT} \label{sec:planning}
In this section, we discuss our approach to find feasible paths to the goal through a backward reachable tree.

\noindent
\textbf{Finding a feasible path}: To find a feasible path to the goal for a query distribution $q \coloneqq (\mu_{q}, \Sigma_{q})$, single hop connections are attempted one-by-one to $M$ nearest nodes on the BRT for some hyperparameter $M$. For a candidate node $\nu_{k}$ on the BRT, the following problem is solved,
\begin{align}
    \mathscr{C}_{q} \longleftarrow \operatorname{OPTSTEER}( (\mu_{q}, \Sigma_{q}), (\mathscr{R}(\mu^{(k)}, \mathscr{P}^{(k)}), \Sigma^{(k)}), N),
\end{align}
where $\operatorname{OPTSTEER}$ is defined in Section~\ref{sec:problem_statement}. The search for a feasible path terminates once a connection has successfully been established to one of the existing nodes on the BRT, and is given by a concatenation of the above computed control sequence $\mathscr{C}_{q}$ with the pre-computed controllers stored in the sequence of edges of the tree from the $k$th node to the root node. Let $ \nu_{k} $ be a distance of $d_{k}$ hops away from the goal s.t.\ $ \mathrm{idx}_{0}, \mathrm{idx}_{1}, \cdots, \mathrm{idx}_{d_{k}}$ be the sequence of nodes encountered from the $k$th node to the root node where $\mathrm{idx}_{0} = k$ and \ $ \mathrm{idx}_{d_{k}} = 0 $. Thererfore, the feasible path from $q$ to the goal $\mathcal{G}$ is obtained as, $ \mathscr{C}_{q, \mathcal{G}} = \mathscr{C}_{q} \bigcup \mathscr{C}_{ \mathrm{idx}_{0}, \mathrm{idx}_{1} } \bigcup \mathscr{C}_{ \mathrm{idx}_{1}, \mathrm{idx}_{2} } \cdots \bigcup \mathscr{C}_{ \mathrm{idx}_{d_{k} - 1}, \mathrm{idx}_{d_{k}} } $ s.t.\ $ q \xrightarrow[(d_{k} + 1)N]{\mathscr{C}_{q, \mathcal{G}}} \mathcal{G}$, which follows from Lemma~\ref{lemma:sequential_composition}.

\noindent
\textbf{Implication of recursive feasibility on the speed-up in computing a feasible path}: Say that the query distribution $q$ is such that $q \in \text{t}-\operatorname{BRS}(\mathcal{G})$ and $q \notin \text{t'}-\operatorname{BRS}(\mathcal{G}) \ \forall \ t' < t$, i.e.\ a path from $q$ to the goal shorter than $t$-hops does not exist. Therefore, to compute a feasible path that steers the system from $q$ to $\mathcal{G}$ without reusing any of the pre-computed controllers from the BRT, a feasibility instance $\operatorname{FEASIBLE}(q, \mathcal{G}, tN)$ of size $tN$ needs to be solved. Alternatively, if the search for a feasible path is carried through attempting connections to the BRT, the expenditure on compute is that of solving a feasibility instance of size $N$ a maximum $M$ number of times resulting in an order of magnitude savings in computation.
\section{Proof of Maximum Coverage of $\operatorname{MAXCOVAR}$ BRT \cite{aggarwal2024sdp_CDC}} \label{sec:proof}
\noindent
In this section, we provide a formal proof of a theorem on maximum coverage of the roadmap technique proposed in \cite{aggarwal2024sdp_CDC}. Ref.~\cite{aggarwal2024sdp_CDC} proposes a novel optimization-based method $\operatorname{MAXCOVAR}$ of node addition and edge controller construction such that the corresponding backward reachable tree provides provably maximal coverage as compared to any other method of constructing a roadmap of distributions. This is formalized through an analysis of the $h$-$\operatorname{BRS}$ of the constructed $\operatorname{MAXCOVAR}$ BRT as Theorem~\ref{theorem:coverage} in the following discussion. We note that Theorem~\ref{theorem:coverage}
is a statement on maximum coverage of $\operatorname{MAXCOVAR}$ BRT \cite{aggarwal2024sdp_CDC} as compared to any method amongst the class of methods that grow a roadmap of distributions \cite{csbrm_tro2024}, and that the proposed approach of this paper $\operatorname{MAXELLIPSOID}$ BRT \textbf{does not} belong to that class since we are concerned with growing a roadmap of ambiguity sets of distributions.

We proceed by first proving Lemma~\ref{lemma:sigma_brs_lemma} that describes the coverage of a distribution (single node on the tree) and is used as a building block for Theorem~\ref{theorem:coverage} that concerns the coverage of a tree (multiple nodes). The $\operatorname{MAXCOVAR}$  optimization program of \cite{aggarwal2024sdp_CDC} is stated below that searches for a maximal covariance and the associated control sequence that steers a query distribution with mean $\mu_{q}$ to a goal distribution $(\mu_{\mathcal{G}}, \Sigma_{\mathcal{G}})$.
\begin{align}
    \min_{ \Sigma_{k}, K_{k}, \mu_{k}, v_{k} } J_{\mathrm{MAX\textunderscore COVAR}} = -\lambda_{\mathrm{min}}(\Sigma_{0}) \tag{MAXCOVAR}\label{eq:J_max_covar}
\end{align}
such that for all $k = 0, 1, \cdots, N - 1$,
\begin{subequations}
\begin{align}
    &\begin{aligned}
        \Sigma_{k+1} = A &\Sigma_{k} A^{\intercal} + B K_{k}\Sigma_{k} A\t + A\Sigma_{k} K\t_{k} B\t \nonumber \\ &+ B K_{k} \Sigma_{k} K\t_{k} B\t + DD\t,
    \end{aligned} \\ 
    &\lambda_{\mathrm{max}}(\Sigma_{N}) \leq \lambda_{\mathrm{min}}(\Sigma_{\mathcal{G}}), \label{eq:covar_goal_maxcovar} \\
    &\mu_{k+1} = A \mu_{k} + B v_{k}, \label{eq:mean_prop_maxcovar} \\
    &\mu_{0} = \mu_{q},\ \ \mu_{N} = \mu_{\mathcal{G}} \label{eq:mean_maxcovar} \\
    &\Phi^{-1}(1-\epsilon_{x}) \sqrt{ \alpha^{\intercal}_{x} \Sigma_{k} \alpha_{x} } + \alpha^{\intercal}_{x} \mu_{k} - \beta_{x} \leq 0, \label{eq:state_constraint_maxcovar} \\
    &\Phi^{-1}(1-\epsilon_{u}) \sqrt{ \alpha\t_{u} K_{k} \Sigma_{k} K\t_{k} \alpha_{u} } + \alpha^{\intercal}_{u} v_{k} - \beta_{u} \leq 0, \label{eq:control_constraint_maxcovar}
\end{align}
\end{subequations}
For the statement of the lemma and the theorem, we recall that a planning scene $ \{ (\alpha_{x}, \beta_{x}, \epsilon_{x}), (\alpha_{u}, \beta_{u}, \epsilon_{u}) \} $ refers to the collection of all parameters that define chance constraints on the state and control input.

\begin{lemma}\label{lemma:sigma_brs_lemma}
    For $\Sigma_{\mathrm{max}}, \Sigma^{-}_{\mathrm{max}} \in \mathbb{S}^{n}_{++}$, s.t.\ $\lambda_{\mathrm{min}}( \Sigma_{\mathrm{max}} ) > \lambda_{\mathrm{min}}( \Sigma^{-}_{\mathrm{max}} )$, $ h\text{-}\operatorname{BRS}( \mu, \Sigma_{\mathrm{max}} ) \supseteq h\text{-}\operatorname{BRS}( \mu, \Sigma^{-}_{\mathrm{max}} ) \ \forall \ h \geq 1 $ for all planning scenes. Also, there exist planning scenes $ \{ (\alpha_{x}, \beta_{x}, \epsilon_{x}), (\alpha_{u}, \beta_{u}, \epsilon_{u}) \} $ s.t.\ $ h\text{-}\operatorname{BRS}( \mu, \Sigma_{\mathrm{max}} ) \supset h\text{-}\operatorname{BRS}( \mu, \Sigma^{-}_{\mathrm{max}} ) \ \forall \ h \geq 1$.
\end{lemma}
\begin{proof}
    See Appendix~\ref{sec:appendix}.
\end{proof}
Before stating Theorem~\ref{theorem:coverage}, we introduce the following notation. Let $ \mathcal{T}^{(r)}_{\mathrm{MAXCOVAR}}$ and $ \mathcal{T}^{(r)}_{\mathrm{ANY}} $ be the state of the BRTs at the $r$th iteration of the node addition procedure following the $\operatorname{MAXCOVAR}$ and $\operatorname{ANY}$ methods of node and edge construction respectively, initialized at the common tree state of a singleton node of the goal distribution. The ANY procedure is any method of node and edge addition that does not explicitly maximize the minimum eigenvalue of the covariance corresponding to any query mean $\mu_{q}$ and goal distribution $(\mu_{\mathcal{G}}, \Sigma_{\mathcal{G}})$. In other words, ANY refers to any algorithm that returns a node covariance and edge controller tuple $(\Sigma_{q}, \mathscr{C}_{q})$ s.t.\ $ \lambda_{\mathrm{min}}(\Sigma_{q}) < \lambda_{\mathrm{min}}(\Sigma_{q, \mathrm{max}}) $ where $ \Sigma_{q, \mathrm{max}}, \mathscr{C}_{q, \mathrm{max}} \longleftarrow \operatorname{MAXCOVAR}(\mu_{q}, (\mu_\mathcal{G}, \Sigma_\mathcal{G}), N)$. ANY could for instance sample the covariance
randomly \cite{csbrm_tro2024}, or add edges that optimize for a regularized objective of the minimum eigenvalue with a performance term.
\begin{theorem}[Maximum Coverage]\label{theorem:coverage}
    $\text{h-}\operatorname{BRS}(\mathcal{T}^{(r)}_{\mathrm{MAXCOVAR}} ) \supseteq \text{h-}\operatorname{BRS}(\mathcal{T}^{(r)}_{\mathrm{ANY}} )$ for all planning scenes, and there always exist planning scenes such that $\text{h-}\operatorname{BRS}(\mathcal{T}^{(r)}_{\mathrm{MAXCOVAR}} ) \supset \text{h-}\operatorname{BRS}(\mathcal{T}^{(r)}_{\mathrm{ANY}} ) \ \forall \ r \geq 1, \ \forall \ h \geq 1$.
\end{theorem}
\begin{proof}
    See Appendix~\ref{sec:appendix}.
\end{proof}
\section{Experiments}\label{sec:experiments}
\begin{figure*}[t]
    \centering
   \includegraphics[width=\textwidth, trim={0 5cm 0 5cm},clip]{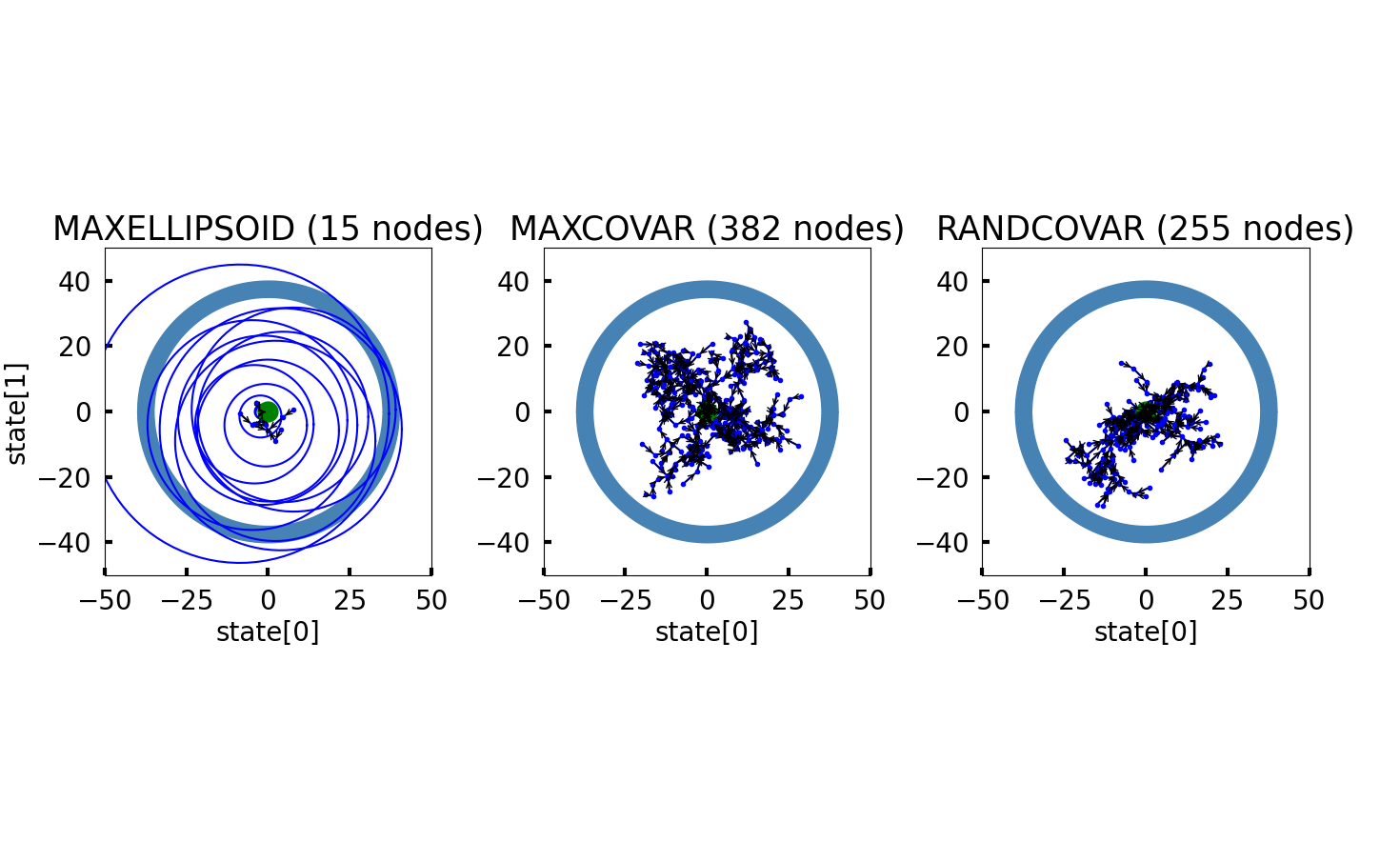}
    \vspace*{-.25in}
    \caption{The above figure show the three backward reachable trees that were used for our planning experiment: $\operatorname{MAXELLIPSOID}$ (15 nodes added in 15 iterations, \textbf{less than 1 min to compute}), $\operatorname{MAXCOVAR}$ (382 nodes added in 1000 iterations, \textbf{around 16 mins to compute}), and $\operatorname{RANDCOVAR}$ (255 nodes added in 3500 iterations, \textbf{more than 1 hour to compute}). For all the three subplots, each blue dot on the plot represents the first two dimensions of the node mean for $\operatorname{MAXCOVAR}$ and $\operatorname{RANDCOVAR}$ trees, and the center of the ellipsoidal mean ambiguity set for the $\operatorname{MAXELLIPSOID}$ tree. Additionally for the $\operatorname{MAXELLIPSOID}$ subplot, Fig.~\ref{fig:three_trees} shows the 2D projections of the 6D ellipsoidal regions corresponding to the mean ambiguity sets stored at each node. The directed edges between the nodes in each of the subplots represent the pre-computed $N$-step control sequences. \textbf{Planning experiment}: Query means were randomly sampled from the blue annulus and a query covariance of $ 0.2\mathbf{I}_{6}$ was used to attempt connects to the trees. \textbf{The proposed $\operatorname{MAXELLIPSOID}$ shows superior performance} as compared to the other two methods as summarized in Table~\ref{tab:exp1} \textbf{with a fraction of the number of nodes (15 nodes)} and \textbf{at a fraction of the time spent computing the tree (less than 1 minute)} as compared to the other two methods.}
    \label{fig:three_trees}
\end{figure*}

To illustrate our method, we conduct experiments for the motion planning of a quadrotor in a 2D plane. The lateral and longitudinal dynamics of the quadrotor are modeled as a triple integrator leading to a 6 DoF model with state matrices,
$$
A=\left[\begin{array}{ccc}
I_2 & \Delta T I_2 & 0_2 \\
0_2 & I_2 & \Delta T I_2 \\
0_2 & 0_2 & I_2
\end{array}\right], ~B=\left[\begin{array}{c}
0_2 \\
0_2 \\
\Delta T I_2
\end{array}\right], ~ D=0.1 I_6,
$$
a time step of 0.1 seconds, a horizon of $N = 20$, and the ambiguity set for the goal distribution $\mathcal{G}$ corresponding to the planning task is as follows:
\begin{align*}
    \mu_{\mathcal{G},c} = \mathbf{0}_{6 \times 1}, \quad \mathscr{P}_{\mathcal{G}} = 0.5 \mathbf{I}_{6\times 6}, \quad \Sigma_{\mathcal{G}} = 0.1  \mathbf{I}_{6\times 6}.
\end{align*}
The control input space is characterized by a bounding box represented as 
$\alpha_u=\left\{\left[
\pm 1,0 \right]^{\top}, 
\left[0,\pm 1 \right]^{\top}\right\}$, 
$\beta_u=\{ \pm 25, \pm 25\}$.
The chance constraint linearization is performed around $ \Sigma_{r} = 1.2 \ \mathbf{I}_{6\times6} $, and $ Y_{r} = 15 \ \mathbf{I}_{2\times2} $. All the optimization programs are solved in Python using cvxpy \cite{cvxpy2016}. We conduct an experiment to compare the two classes of trees: backward reachable tree (BRT) of distributions: $\operatorname{MAXCOVAR}$ and $\operatorname{RANDCOVAR}$ BRTs, and a BRT of ambiguity sets of distributions, $\operatorname{MAXELLIPSOID}$ BRT.
\noindent
\subsection{Construction of the BRTs}\label{sec:construct_brt}
For the tree construction procedure, a sampling radius of $r_{\mathrm{sample}} = \left[ \pm 5, \pm 5, \pm 2.5, \pm 2.5, \pm 1.25, \pm 1.25 \right] $ was used across all tree types. Fig.~\ref{fig:three_trees} shows the three types of trees that were constructed: $\operatorname{MAXELLIPSOID}$ (followed by a $\operatorname{MAXCOVARELL}$ procedure for each edge as a part of the bi-level search for the maximally distributionally robust controller), $\operatorname{MAXCOVAR}$ and the $\operatorname{RANDCOVAR}$ BRT. $\operatorname{MAXCOVAR}$ \cite{aggarwal2024sdp} \cite{aggarwal2024sdp_CDC} and $\operatorname{RANDCOVAR}$ are trees of distributions where each edge is a covariance steering controller \cite{liu2022optimal} \cite{rapakoulias2023discrete}. All the trees were constructed using the same seed.

Fig.~\ref{fig:three_trees} displays the first two states of the 6 DoF model: the x ($\operatorname{state}[0]$) and y ($\operatorname{state}[1]$) locations of the quadrotor in the x-y plane. For all the three subplots, each node (blue dot) on the plot represents the first two dimensions of the state distribution mean for $\operatorname{MAXCOVAR}$ and $\operatorname{RANDCOVAR}$ trees, and the center of the ellipsoidal first-order moment ambiguity sets for the $\operatorname{MAXELLIPSOID}$ tree. Additionally for the $\operatorname{MAXELLIPSOID}$ subplot, Fig.~\ref{fig:three_trees} shows the 2D projections of the 6D ellipsoidal regions corresponding to the first-order moment ambiguity sets stored at each node. The directed edges between the nodes in each of the subplots are the corresponding $N$-step control sequences that are stored offline.

The node addition procedure for the $\operatorname{MAXELLIPSOID}$ BRT was run for 50 iterations and 15 nodes were added to the tree. For the $\operatorname{MAXCOVAR}$ BRT, 382 nodes were added in 1000 iterations, and for the $\operatorname{RANDCOVAR}$ BRT, 255 nodes were added in 3500 iterations.

The construction procedure for the MAXELLIPSOID BRT was followed as described in Algorithm~\ref{alg:construct_brt}. For each edge, an initial covariance of $ \Sigma_{q} = 0.1 \mathbf{I}_{6}$ was used, followed by a $\operatorname{MAXCOVARELL}$ optimization step (see   remark in Section~\ref{sec:bilevel}). 

The $\operatorname{MAXCOVAR}$ tree was constructed following an algorithm with similar sub-routines, the difference being in the edge controller which is a covariance steering controller between nodes representing distributions (see \cite{aggarwal2024sdp} \cite{aggarwal2024sdp_CDC} for details), instead of a distributionally robust controller between ambiguity sets. Finally, for the $\operatorname{RANDCOVAR}$ tree, let $\nu_{k}$ be the node on the tree that's selected to expand and let $\mu_{\mathrm{cand.}}$ be the candidate mean sampled from a box around it. The node covariance $\Sigma_{\mathrm{cand.}, \mathrm{rand}}$ is randomly sampled from the space of positive definite matrices, similar to \cite{csbrm_tro2024}, and the edge controller is given by the optimal steering control from $(\mu_{\mathrm{cand.}}, \Sigma_{\mathrm{cand.}, \mathrm{rand}})$ to $ (\mu_{k}, \Sigma_{k})$. To construct samples from the positive definite matrix space, the eigenvalues and orthonormal eigenvectors that constitute a positive definite matrix are sampled separately. It was observed empirically that randomly sampling node covariances resulted in rejecting a lot of candidate nodes due to the non-existence of a corresponding steering maneuver. Therefore, eigenvalues of the candidate node covariance $\Sigma_{\mathrm{cand., rand}}$ were sampled to ensure that \ $ \lambda_{\mathrm{min}}( \Sigma_{\mathrm{cand.}, \mathrm{rand}}) \leq \lambda_{\mathrm{min}}( \Sigma_{\mathrm{cand.}, \mathrm{max}})$ where $ \Sigma_{\mathrm{cand.}, \mathrm{max}}$ is the computed maximal covariance corresponding to the sampled $\mu_{\mathrm{cand.}}$ (computation of this maximal covariance is detailed in \cite{aggarwal2024sdp} \cite{aggarwal2024sdp_CDC}).

We now proceed to describe our planning experiment. The experiment tries to find paths to the goal by attempting to connect to the trees for sampled query distributions (query mean and covariance are sampled independently). The x-y dimensions of the query means were sampled from the blue annulus in the x-y plane (see Fig.~\ref{fig:three_trees}). Readers are referred to Section~\ref{sec:exp1} for more details on sampling the query mean including the sampling of the velocity and acceleration space. The experiment is aimed at demonstrating the space-filling property of the $\operatorname{MAXELLIPSOID}$ tree as compared to $\operatorname{MAXCOVAR}$ and $\operatorname{RANDCOVAR}$ assuming a fixed query covariance of $0.2\mathbf{I}_{6}$.
\noindent
\subsection{Experiment}\label{sec:exp1}
\textbf{Setup}: 
The following radii of the annulus were used for sampling the x-y space of the query mean: $r_\mathrm{inner} = 35$, $r_\mathrm{outer} = 40$ (see Fig.~\ref{fig:three_trees}), and the query covariance was $ \Sigma_{q} = 0.2 \mathbf{I}_{6}$ across all runs. The velocity space coordinates (third and fourth dimension) were sampled independently and uniformly from the interval $(-2.5, 2.5)$ and the acceleration space coordinates (fifth and sixth dimension) were sampled uniformly and independently from the interval $(-0.625, 0.625)$.
The experiment was repeated 250 times and the number of times a path was successfully found was recorded (see Table~\ref{tab:exp1}).

\begin{table}[t]
  \centering
  \caption{Planning Experiment: No. of successful paths found}
  \label{tab:exp1}
  \begin{tabular}{|c|c|c|}
    \hline
     \textbf{MAXELLIPSOID} & \textbf{MAXCOVAR} & \textbf{RANDCOVAR} \\
     with 15 nodes & with 382 nodes & with 255 nodes \\
    \hline
     249 out of 250 & 40 out of 250 & 40 out of 250  \\
    \hline
  \end{tabular}
\end{table}

\textbf{Interpretation}: Table~\ref{tab:exp1} shows the superior performance of $\operatorname{MAXELLIPSOID}$ BRT compared to $\operatorname{MAXCOVAR}$ and $\operatorname{RANDCOVAR}$ BRTs in finding paths with only a fraction of the number of nodes (15 nodes added in 50 iterations against 382 nodes added in 1000 iterations and 255 nodes added in 3500 iterations respectively). A compact size of the library of controllers is beneficial for practical purposes since constructing the $\operatorname{MAXELLIPSOID}$ BRT took less than 1 min, while constructing the $\operatorname{MAXCOVAR}$ and $\operatorname{RANDCOVAR}$ BRTs took approximately 16 min and more than 1 hour respectively.

\begin{figure}[t]
    \centering
    \includegraphics[trim=220 0 250 50, clip, width=0.8\columnwidth]{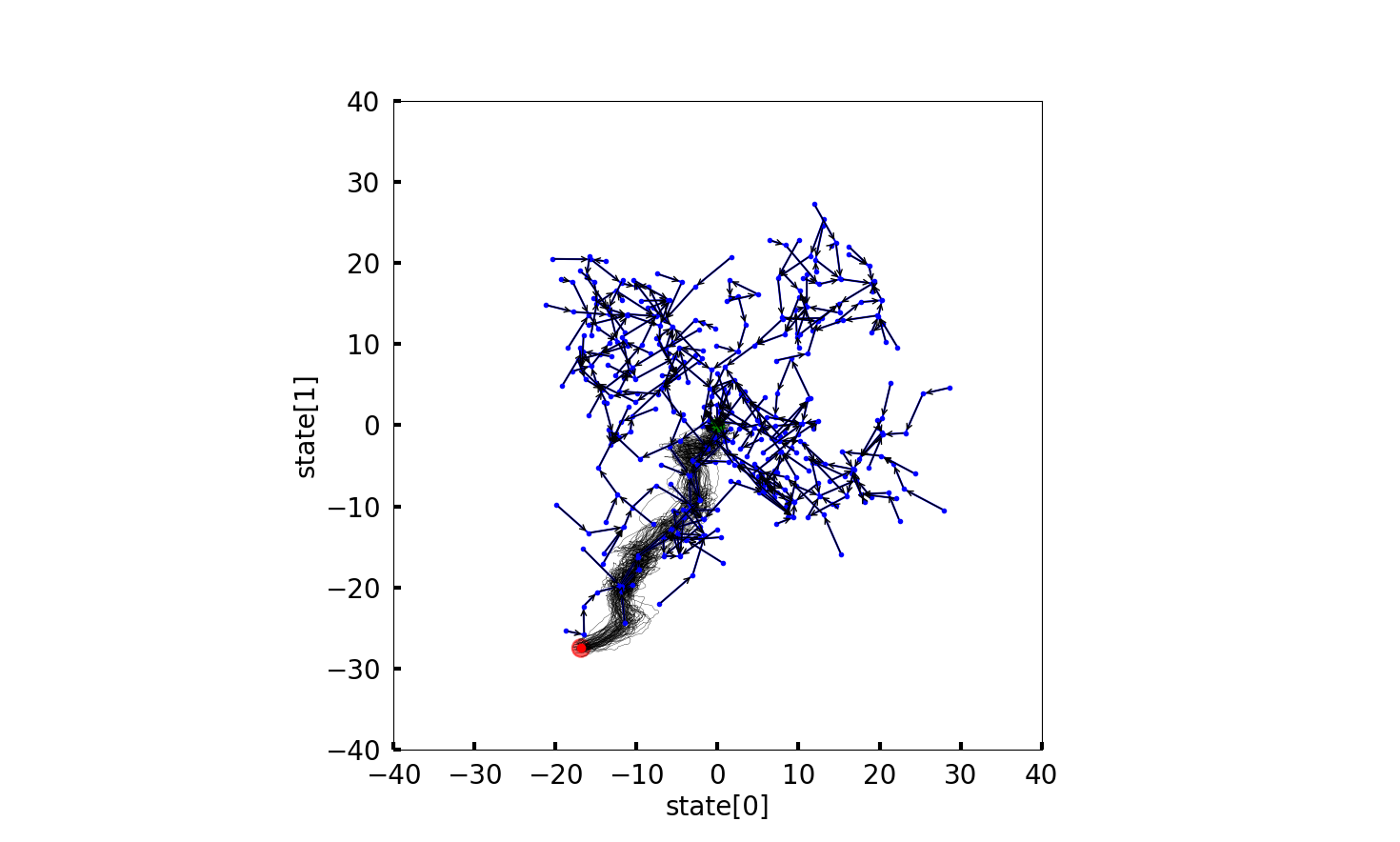}
    \vspace*{-.2in} 
    \label{fig:maxcovar_brt_mc}
%
    \centering
    \includegraphics[trim=220 0 250 50, clip, width=0.8\columnwidth]{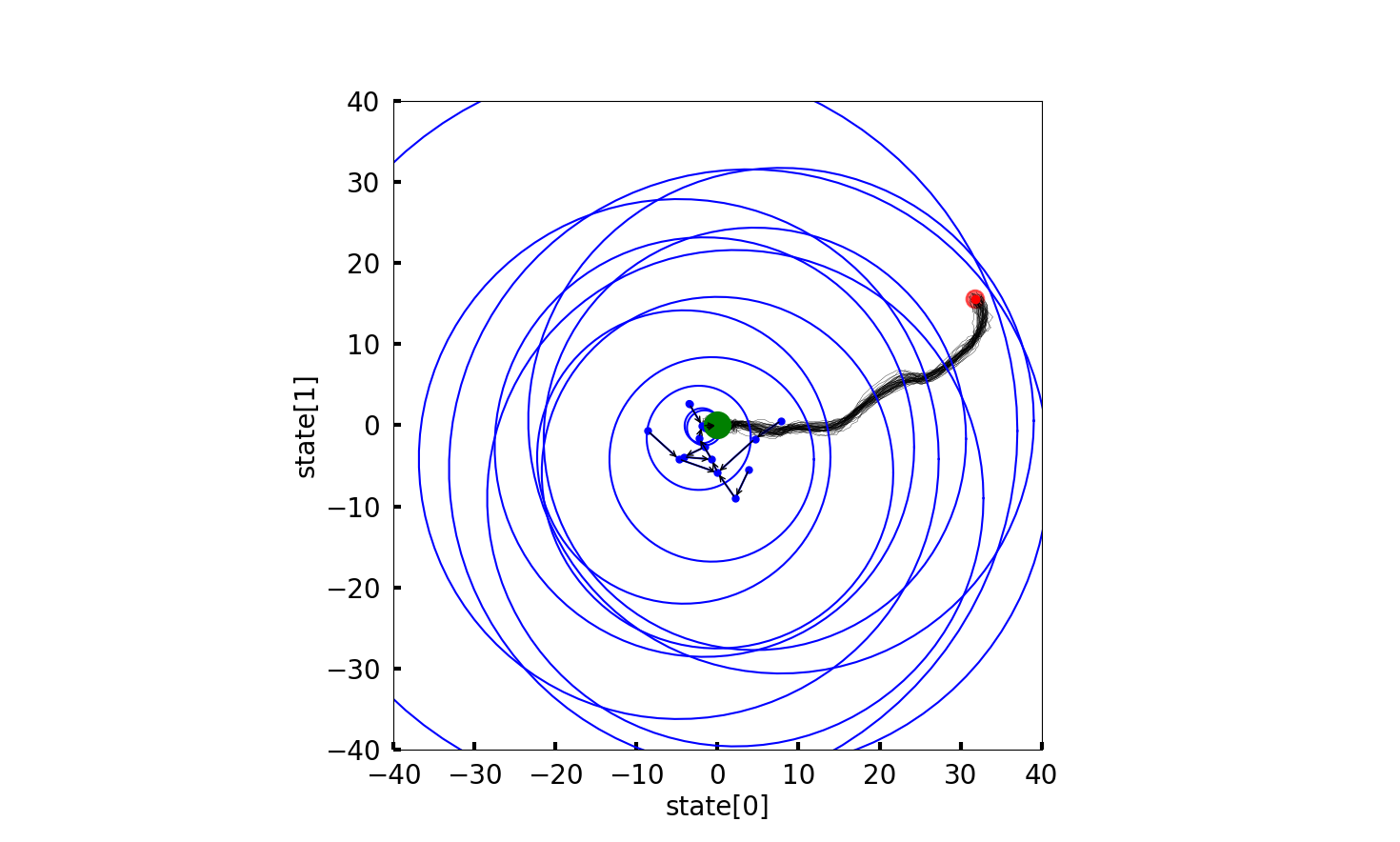}
    \vspace*{-.2in} 
    \caption{Real-time planning to the goal (origin) through the precomputed $\operatorname{MAXCOVAR}$ and $\operatorname{MAXELLIPSOID}$ BRTs respectively for a randomly sampled query distribution (red ellipse corresponds to the 3-$\sigma$ confidence interval of the sampled query distribution). The only real-time computation was the steering control sequence from the query distribution to the first node on the branch of the tree corresponding to the discovered feasible path. Black lines represent Monte-Carlo trajectories for the system evolution under the concatenated control sequence.}
    \label{fig:maxellips_brt_mc}
\end{figure}
\textbf{Conclusion}: Growing a backward reachable tree of ambiguity sets and the corresponding distributionally robust controllers has a space-filling property and is a computationally efficient method to populate a pre-computed library of controllers for multi-query planning. $\operatorname{MAXELLIPSOID}$ BRT has a compact representation, is faster to compute and performs drastically better as compared to $\operatorname{MAXCOVAR}$ and $\operatorname{RANDCOVAR}$ BRTs as demonstrated in an example planning scenario.
\subsection{Real-time planning through the BRTs} Fig.~\ref{fig:maxellips_brt_mc} shows the Monte Carlo trajectories (represented by black lines) corresponding to a randomly sampled query distribution for the $\operatorname{MAXCOVAR}$ and the $\operatorname{MAXELLIPSOID}$ BRTs respectively. The red ellipse corresponds to the 3-$\sigma$ confidence interval of the sampled query node, and the green ellipse refers to the goal node (ellipsoidal region corresponding to the goal mean ambiguity set overlaid with the 3-$\sigma$ confidence ellipse corresponding to the goal covariance). The only real-time computation done was to compute a control sequence that steers the system from the query distribution to one of the existing nodes on the trees, trajectories for the remaining time-steps were propagated using controllers stored along each edge of the discovered feasible path through the BRTs.

\appendix \label{sec:appendix}

\begin{proof}[\textbf{Proof of Theorem~\ref{theorem:relaxation}}]
    It has been proved in the discussion in Section~\ref{sec:convex_relax} that (\ref{eq:J_fd})--(\ref{eq:J_fh}) are sufficient conditions for (\ref{eq:mean_init_nlp})--(\ref{eq:state_constraint_nlp}) to hold. We now argue that the decision variables $U_{k}$, $Y_{k}$, and $\Sigma_{k}$ obtained as a solution to (\ref{eq:J_final}) recover a feasible control sequence for the nonlinear program (\ref{eq:J_simp}) such that (\ref{eq:J_simp}a)--(\ref{eq:J_simp}c) and (\ref{eq:J_simp}i) hold. Specifically, we show that under the control sequence obtained as a solution to the relaxed convex semidefinite program (\ref{eq:J_final}), the propagated covariance dynamics initialized at $ \Sigma_{0} = \Sigma_{q} $ (\ref{eq:J_simp}b) are such that the terminal covariance constraint (\ref{eq:J_simp}c) and control chance constraints (\ref{eq:J_simp}i) hold.

Let $\{ \tilde{U}_{k}, \tilde{Y}_{k}, \tilde{\Sigma}_{k} \}_{k=0}^{N-1}$ be the solution to the relaxed program (\ref{eq:J_final}) such that $\mathscr{C} \coloneqq \{ \mathscr{C}_{k} \}_{k=0}^{N-1}$ denotes the recovered control sequence where $ \mathscr{C}_{k} = \{ \tilde{K}_{k}, \tilde{v}_{k} \}$ and $ \tilde{K}_{k} = \tilde{U}_{k} \tilde{\Sigma}
\inv_{k} $. Let $\{\Sigma_{k}\}_{k=0}^{N-1}$ represent the covariance dynamics of the system initialized at $\Sigma_{0} = \Sigma_{q}$ (\ref{eq:J_simp}b) and steered through the feedback control sequence $ \{\tilde{K}_{k}\}_{k=0}^{N-1}$. Note that we make the distinction between $\Sigma_{k}$ and $\tilde{\Sigma}_{k}$ where the former represents the actual covariance dynamics (\ref{eq:J_simp}a) steered through the recovered feasible control sequence and the latter refers to a decision variable of the relaxed program (\ref{eq:J_final}). We make this explicit through the notation $ \Sigma_{k} \coloneqq \Sigma_{k}(\Sigma_{0}, \mathscr{C})$ s.t. $\Sigma_{k+1} = A \Sigma_{k} A^{\intercal} + B \tilde{K}_{k}\Sigma_{k} A\t + A\Sigma_{k} \tilde{K}\t_{k} B\t \nonumber + B \tilde{K}_{k} \Sigma_{k} \tilde{K}\t_{k} B\t + DD\t$.

We show that the theorem holds by proving $ \Sigma_{k}(\Sigma_{0}, \mathscr{C}) \preceq \tilde{\Sigma}_{k} \ \forall \ k = 0, 1, \cdots, N-1$. Assume that $\Sigma_{k-1} (\Sigma_{0}, \mathscr{C}) \preceq \tilde{\Sigma}_{k-1}$ holds. From (\ref{eq:J_simp}a),
\begin{align*}
    \Sigma_{k}(\Sigma_{0}, \mathscr{C}) &= (A + B \tilde{K}_{k-1}) \Sigma_{k-1}(\Sigma_{0}, \mathscr{C}) (A + B \tilde{K}_{k-1})\t + D D\t \\
    &\preceq (A + B \tilde{K}_{k-1}) \tilde{\Sigma}_{k-1}(A + B \tilde{K}_{k-1})\t + D D\t \\
    &=  A \tilde{\Sigma}_{k-1} A^{\intercal} + B \tilde{K}_{k-1}\tilde{\Sigma}_{k-1} A\t + A\tilde{\Sigma}_{k-1} \tilde{K}\t_{k-1} B\t \\ &+ B \tilde{K}_{k-1} \tilde{\Sigma}_{k-1} \tilde{K}\t_{k-1} B\t + DD\t
\end{align*}
From (\ref{eq:J_fa}), $ \tilde{K}_{k-1} \tilde{\Sigma}_{k-1} \tilde{K}\t_{k-1} = \tilde{U}_{k-1} \tilde{\Sigma}\inv_{k-1} \tilde{U}\t_{k-1} \preceq \tilde{Y}_{k-1}$. Substituting $ \tilde{U}_{k-1} = \tilde{K}_{k-1} \tilde{\Sigma}_{k-1}$,
\begin{align*}
\Sigma_{k}(\Sigma_{0}, \mathscr{C}) &\preceq A \tilde{\Sigma}_{k-1} A^{\intercal} + B \tilde{U}_{k-1} A\t + A\tilde{U}\t_{k-1} B\t \\ &+ B \tilde{Y}_{k-1} B\t + DD\t 
\end{align*}
From (\ref{eq:J_final}b), $G_{k}(\tilde{\Sigma}_{k}, \tilde{\Sigma}_{k-1}, \tilde{Y}_{k}, \tilde{U}_{k}) = 0$, we have $ \tilde{\Sigma}_{k} =  A\tilde{\Sigma}_{k-1}A\t + B \tilde{U}_{k-1} A\t + A \tilde{U}\t_{k-1} B\t + B \tilde{Y}_{k-1} B\t + D D\t$. Therefore, $ \Sigma_{k}(\Sigma_{0}, \mathscr{C}) \preceq \tilde{\Sigma}_{k}$ holds if $ \Sigma_{k-1}(\Sigma_{0}, \mathscr{C}) \preceq \tilde{\Sigma}_{k-1}  $. Since $\Sigma_{0}(\Sigma_{0}, \hat{K}_{k}) = \tilde{\Sigma}_{0} = \Sigma_{q}$, from induction $ \Sigma_{k}(\Sigma_{0}, \hat{K}_{k}) \preceq \tilde{\Sigma}_{k} \ \forall \ k = 0, 1, \cdots, N-1$.
From (\ref{eq:J_final}c), $ \lambda_{\mathrm{max}} (\tilde{\Sigma}_{N}) \leq \lambda_{\mathrm{min}}(\Sigma_{\mathcal{G}})$. Since $ \Sigma_{N} \coloneqq \Sigma_{N}(\Sigma_{0}, \mathscr{C}) \preceq \tilde{\Sigma}_{N}$, therefore $ \lambda_{\mathrm{max}}(\Sigma_{N}) \leq \lambda_{\mathrm{max}}(\tilde{\Sigma}_{N})$ (follows from: for $A, B \in \mathbb{S}^{n}_{++}$, $A \preceq B \implies \lambda_{i}(A) \leq \lambda_{i}(B)$ where $\lambda_{i}(.)$ is the $i$-th largest eigenvalue). It follows that $ \lambda_{\mathrm{max}} (\Sigma_{N}) \leq \lambda_{\mathrm{max}} (\tilde{\Sigma}_{N}) \leq \lambda_{\mathrm{min}}(\Sigma_{\mathcal{G}}) $ and hence (\ref{eq:J_simp}c) holds. We can similarly show that (\ref{eq:J_simp}i) also holds and the theorem follows.
\end{proof}

\begin{proof}[\textbf{Proof of Lemma}~\ref{lemma:sigma_brs_lemma}]
    We want to show that $ \ \forall \ p \in h\text{-}\operatorname{BRS}(\mu, \Sigma^{-}_{\mathrm{max}}); \ \exists \ \mathscr{C}_{p} $ s.t.\ $ p \xlongrightarrow[hN]{\mathscr{C}_{p}} (\mu, \Sigma_{\mathrm{max}}) $, what we refer to as the \textit{forward side} of the argument, and also that $\ \exists \ q $ s.t.\ $ q \in h\text{-}\operatorname{BRS}(\mu, \Sigma_{\mathrm{max}})$ and $ \nexists \ \mathscr{C}_{q} $ s.t.\ $ q \xlongrightarrow[hN]{\mathscr{C}_{q}} (\mu, \Sigma^{-}_{\mathrm{max}}) $, what we refer to as the \textit{backward side} of the argument.

    We first prove the forwards side of the argument. Let $p$ be any element of $h\text{-}\operatorname{BRS}(\mu, \Sigma^{-}_{\mathrm{max}})$. Therefore, $ \exists \ \mathscr{C}_p $ s.t.\ (\ref{eq:covar_goal_maxcovar})--(\ref{eq:control_constraint_maxcovar}) are satisfied. From (\ref{eq:covar_goal_maxcovar}), $ \Sigma_{{hN}} (p, \mathscr{C}_{p}) $ s.t.\ $ \lambda_{\mathrm{max}}( \Sigma_{hN} (p, \mathscr{C}_{p}) ) \leq \lambda_{\mathrm{min}}(\Sigma^{-}_{\mathrm{max}})$. Since $ \lambda_{\mathrm{min}}( \Sigma^{-}_{\mathrm{max}} ) < \lambda_{\mathrm{min}}( \Sigma_{\mathrm{max}} ) $, $ \lambda_{\mathrm{max}}( \Sigma_{hN} (p, \mathscr{C}_{p}) ) \leq \lambda_{\mathrm{min}}(\Sigma_{\mathrm{max}})$. Therefore, $ p \in h\text{-}\operatorname{BRS}(\mu, \Sigma_{\mathrm{max}})$ under the same control sequence $ \mathscr{C}_{p} $ that drives $p$ to $ (\mu, \Sigma^{-}_{\mathrm{max}}) $. Because the above argument holds for any arbitrary distribution that reaches $(\mu, \Sigma^{-}_{\mathrm{max}})$, it follows that $ \forall \ p \in h\text{-}\operatorname{BRS}(\mu, \Sigma^{-}_{\mathrm{max}}), p \in h\text{-}\operatorname{BRS}(\mu, \Sigma_{\mathrm{max}})$.

    Now, we prove the reverse -- specifically, we construct $q$ and a planning scene $ \{ (\alpha_{x}, \beta_{x}, \epsilon_{x}), (\alpha_{u}, \beta_{u}, \epsilon_{u}) \} $ such that $q \in h\text{-}\operatorname{BRS}(\mu, \Sigma_{\mathrm{max}})$ and $ q \notin h\text{-}\operatorname{BRS}(\mu, \Sigma^{-}_{\mathrm{max}}) $. Consider $q(\mu_{q}, \Sigma_{q})$ such that,
    \begin{align}
        \Sigma_{q}, \mathscr{C}_{q} \longleftarrow \operatorname{MAX-COVAR}( \mu_{q}, (\mu, \Sigma^{-}_{\mathrm{max}}), hN), \label{eq:sigma_q}
    \end{align}
    and $\lambda_{\mathrm{max}}(\Sigma_{q}) < \infty$. We also assume that for $ (\mu_{q}, \Sigma_{q})$ and $\mathscr{C}_{q}$, the corresponding state and control chance constraints (\ref{eq:state_constraint_maxcovar})--(\ref{eq:control_constraint_maxcovar}) are non-tight, i.e., are strict inequalities. This is a mild assumption and can be shown to hold by constructing $\mu_{q}$ in an appropriate manner.

    From (\ref{eq:sigma_q}), $\nexists \ \Sigma_{q^{+}}, \mathscr{C}_{q^{+}}$ s.t.\ $ \lambda_{\mathrm{min}}(\Sigma_{q^{+}}) > \lambda_{\mathrm{min}}(\Sigma_q)$ and 
    $ (\mu_q, \Sigma_{q^{+}}) \xlongrightarrow[kN]{\mathscr{C}_{q^{+}}} (\mu, \Sigma^{-}_{\mathrm{max}})$ $\ast$\label{sent:perturb}.
    
    Now, consider a perturbation to $q$ such that $\Sigma(\epsilon) = \Sigma_{q} + \epsilon \mathrm{I}$ for some $\epsilon > 0$. From $(\ast)$~\ref{sent:perturb}, $ (\mu_{q}, \Sigma(\epsilon)) \notin h\text{-}\operatorname{BRS}(\mu, \Sigma^{-}_{\mathrm{max}})$ since $ \lambda_{\mathrm{min}}(\Sigma(\epsilon)) > \lambda_{\mathrm{min}}(\Sigma_{q})$ for $\epsilon > 0$. We show that $\exists$ planning scenes $\{ (\alpha_{x}, \beta_{x}, \epsilon_{x}), (\alpha_{u}, \beta_{u}, \epsilon_{u}) \}$ $\exists \ \epsilon > 0$ s.t.\ $ (\mu_q, \Sigma(\epsilon)) \in h\text{-}\operatorname{BRS}(\mu, \Sigma_{\mathrm{max}})$.

    Consider $\mathscr{C}(\lambda)$ as the candidate control sequence s.t.\ $ (\mu_{q}, \Sigma(\epsilon)) \xlongrightarrow[hN]{\mathscr{C}(\lambda)} (\mu, \Sigma_{\mathrm{max}})$ where $\lambda \coloneqq \{\lambda_{k}\}_{t=0}^{hN-1}$ is a perturbation to $\mathscr{C}_{q}$ defined as follows,
    \begin{align*}
        K_{k}(\mathscr{C}(\lambda)) &= (1-\lambda_{k}) K_{k}(\mathscr{C}_{q}), \\
        v_{k}(\mathscr{C}(\lambda)) &= v_{k}(\mathscr{C}_{q}),
    \end{align*}
    $\forall \ k = 0, 1, \cdots hN-1.$ Under the perturbed control sequence $\mathscr{C}(\lambda)$, the mean dynamics are unaffected since the feedforward term $v_k(\mathscr{\lambda}) = v_{k}(\mathscr{C}_{q})$ is unperturbed, i.e.\ $\mu_{k}(\mathscr{C}(\lambda)) = \mu_{k}(\mathscr{C}_{q})$. We now express the dynamics of the state covariance under the perturbed initial covariance and control sequence $\Sigma(\epsilon)$ and $\mathscr{C}(\lambda)$ as $ \Sigma_{k}(\Sigma(\epsilon), \mathscr{C}(\lambda))$, yielding
    \begin{align*}
        \Sigma_{1}(\Sigma, \mathscr{C}) &= (A + BK_{0}(\mathscr{C}))(\Sigma_{q} + \epsilon \mathrm{I}) (A + BK_{0}(\mathscr{C}))\t
        + DD\t  \\
        &= \Sigma_{1}(\Sigma_{q}, \mathscr{C}) + \epsilon \gamma_{1}(\mathscr{C})
    \end{align*}
    where $ \Sigma_{1}(\Sigma_{q}, \mathscr{C}) = (A + BK_{0}(\mathscr{C}))\Sigma_{q}(A + BK_{0}(\mathscr{C}))\t $ and $ \gamma_{1}(\mathscr{C}) = (A + BK_{0}(\mathscr{C}))(A + BK_{0}(\mathscr{C}))\t$. By induction, we can express $ \Sigma_{k}(\Sigma, \mathscr{C})$ as,
    \begin{align}
        \Sigma_{k}(\Sigma, \mathscr{C}) = \Sigma_{k}(\Sigma_{q}, \mathscr{C}) + \epsilon \gamma_{k}(\mathscr{C}), \label{eq:sigmak_recurrence}
    \end{align}
    where $ \Sigma_{0}(\Sigma_{q}, \mathscr{C}) = \Sigma_{q}$, $ \gamma_{0}(\mathscr{C}) = \mathrm{I}$, and 
    $$ \gamma_{k}(\mathscr{C}) = (A + BK_{k-1}(\mathscr{C}))\gamma_{k-1}(\mathscr{C})(A + BK_{k-1}(\mathscr{C}))\t . $$
    Writing the state chance constraints (\ref{eq:state_constraint_maxcovar}) corresponding to $ \Sigma(\epsilon) $ and $ K(\mathscr{C}(\lambda))$,
    \begin{align}
        \mathscr{S}_{k}(\Sigma, \mathscr{C}) &= \Phi^{-1}(1-\epsilon_{x})\sqrt{ \alpha_{x}\t \Sigma_{k}(\Sigma, \mathscr{C})\alpha_{x}} + \alpha_{x}\t \mu_{k}(\mathscr{C} ) - \beta_{x}. \label{eq:state_chanceconstraint_perturb1}
    \end{align}
    From (\ref{eq:sigmak_recurrence}) and (\ref{eq:state_chanceconstraint_perturb1}), and using $\sqrt{a + b} \leq \sqrt{a} + \sqrt{b} \ \forall \ a,b \geq 0$,
    \begin{align}
        \mathscr{S}_{k}(\Sigma, \mathscr{C}) &\leq  \Phi^{-1}(1-\epsilon_{x}) \left[\sqrt{ \alpha_{x}\t \Sigma_{k}(\Sigma_{q}, \mathscr{C})\alpha_{x}} \right. \nonumber \\ &+ \left. \sqrt{\epsilon}\sqrt{ \alpha_{x}\t \gamma_{k}(\mathscr{C})\alpha_{x}}  \right] 
         + \alpha_{x}\t \mu_{k}(\mathscr{C}) - \beta_{x}. \label{eq:state_chanceconstraint_perturb}
    \end{align}
    From the state chance constraint corresponding to $\Sigma_{q}$ and $\mathscr{C}_{q}$,
    \begin{align}
        \mathscr{S}_{k}(\Sigma_{q}, \mathscr{C}_{q}) = &\Phi^{-1}(1-\epsilon_{x})\sqrt{ \alpha_{x}\t \Sigma_{k}(\Sigma_{q}, \mathscr{C}_{q})\alpha_{x}} + \alpha_{x}\t \mu_{k}(\mathscr{C}_{q}) \nonumber \\ &- \beta_{x} < 0.
    \end{align}
    Therefore,
    \begin{align}
        \alpha_{x}\t \mu_{k}(\mathscr{C}_{q}) - \beta_{x} = -z_{s,k} - \left[ \Phi^{-1}(1-\epsilon_{x})\sqrt{ \alpha_{x}\t \Sigma_{k}(\Sigma_{q}, \mathscr{C}_{q})\alpha_{x}} \right], \label{eq:state_constraint_slack_substitute}
    \end{align}
    where $z_{s,k} > 0$ is the slack variable associated with the state chance constraint at the $k$-th time-step for the system initialized at $\Sigma_{q}$ and evolving under $\mathscr{C}_{q}$. Substituting (\ref{eq:state_constraint_slack_substitute}) in (\ref{eq:state_chanceconstraint_perturb}), and since $ \mu_{k}(\mathscr{C}) = \mu_{k}(\mathscr{C}_{q})$,
    \begin{align}
        \mathscr{S}_{k}(\Sigma, \mathscr{C}) &\leq  \Phi^{-1}(1-\epsilon_{x}) \left[\sqrt{ \alpha_{x}\t \Sigma_{k}(\Sigma_{q}, \mathscr{C})\alpha_{x}} \right. \nonumber \\ & -\left. \sqrt{ \alpha_{x}\t \Sigma_{k}(\Sigma_{q}, \mathscr{C}_{q})\alpha_{x}}  \right] + \sqrt{\epsilon}\sqrt{ \alpha_{x}\t \gamma_{k}(\mathscr{C})\alpha_{x}} - z_{s,k}. \label{eq:state_chanceconstraint_perturb_1}        
    \end{align}
    We construct perturbations $\epsilon$ and $\lambda$ s.t.\ $ \mathscr{S}_{k}(\Sigma, \mathscr{C}) \leq 0 \ \forall \ k = 0, 1, \ldots, hN-1$. The following holds for the desired value of the perturbations,
    \begin{align}
        \sqrt{\epsilon}\sqrt{ \alpha_{x}\t \gamma_{k}(\mathscr{C})\alpha_{x}} \leq &z_{s,k} - \Phi^{-1}(1-\epsilon_{x}) \left[ \sqrt{ \alpha_{x}\t \Sigma_{k}(\Sigma_{q}, \mathscr{C})\alpha_{x}} \right. \nonumber \\  &- \left. \sqrt{ \alpha_{x}\t \Sigma_{k}(\Sigma_{q}, \mathscr{C}_{q})\alpha_{x}} \right]. \tag{S}\label{eq:state_constraint_ineq}
    \end{align}
    Following a similar analysis as the state chance constraints above, we can express the control chance constraints (\ref{eq:control_constraint_maxcovar}) corresponding to $\Sigma(\epsilon)$ and $\mathscr{C}(\lambda)$ as,
    \begin{align}
        \mathcal{U}_{k}(\Sigma, \mathscr{C}) &\leq  \Phi^{-1}(1-\epsilon_{u}) \left[\sqrt{ \alpha_{u}\t K_k(\mathscr{C}) \Sigma_{k}(\Sigma_{q}, \mathscr{C})K_k\t(\mathscr{C})\alpha_{u}} \right. \nonumber \\
        & \left. -\sqrt{ \alpha_{u}\t K_k(\mathscr{C})\Sigma_{k}(\Sigma_{q}, \mathscr{C}_{q})K_k\t(\mathscr{C})\alpha_{u}}  \right] \nonumber \\ &+ \sqrt{\epsilon}\sqrt{ \alpha_{u}\t K_k(\mathscr{C}) \gamma_{k}(\mathscr{C}) K_k\t(\mathscr{C}) \alpha_{u}} - z_{u,k}, \label{eq:control_chanceconstraint_perturb}
    \end{align}
    where $z_{u,k} > 0$ is the slack variable associated with the control chance constraint at the $k$-th time-step. For the desired values of the perturbations, $ \mathcal{U}_{k}(\Sigma, \mathscr{C}) \leq 0$, and we get the following,
    \begin{align}
        \sqrt{\epsilon}&\sqrt{ \alpha_{u}\t K_k(\mathscr{C}) \gamma_{k}(\mathscr{C}) K_k\t(\mathscr{C}) \alpha_{u}} \leq  \nonumber \\ & \Phi^{-1}(1-\epsilon_{u}) \left[ \sqrt{ \alpha_{u}\t K_k(\mathscr{C})\Sigma_{k}(\Sigma_{q}, \mathscr{C}_{q})K_k\t(\mathscr{C})\alpha_{u}} \right. \nonumber \\ &- \left. \sqrt{ \alpha_{u}\t K_k(\mathscr{C}) \Sigma_{k}(\Sigma_{q}, \mathscr{C})K_k\t(\mathscr{C})\alpha_{u}}  \right]  + z_{u,k}. \tag{U} \label{eq:control_constraint_ineq}
    \end{align}
    For there to always exist some $\epsilon > 0$ s.t.\ (\ref{eq:state_constraint_ineq}) and (\ref{eq:control_constraint_ineq}) hold true, it is a sufficient condition that $ \sqrt{ \alpha_{x}\t \Sigma_{k}(\Sigma_{q}, \mathscr{C}_{q})\alpha_{x}} - \sqrt{ \alpha_{x}\t \Sigma_{k}(\Sigma_{q}, \mathscr{C}(\lambda))\alpha_{x}} \geq 0$, and $ \sqrt{ \alpha_{u}\t K_{k}(\mathscr{C}) \Sigma_{k}(\Sigma_{q}, \mathscr{C}_{q}) K\t_{k}(\mathscr{C}) \alpha_{u}} - \sqrt{ \alpha_{u}\t K_{k} \Sigma_{k}(\Sigma_{q}, \mathscr{C}(\lambda)) K\t_{k} \alpha_{u}} \geq 0$, since $ z_{s,k}$ and $z_{u,k}$ are strictly positive $ \forall \ k = 0, 1, \cdots, hN-1$. Therefore, we construct $\lambda$ such that $ \Sigma_{k}(\Sigma_{q}, \mathscr{C}(\lambda)) \prec \Sigma_{k}(\Sigma_{q}, \mathscr{C}_{q}) \ \forall \ k=0, 1, \cdots, hN-1 $. Writing $\Sigma_{1}(\Sigma_{q}, \mathscr{C}(\lambda))$ in terms of $\Sigma_{1}(\Sigma_{q},
    \mathscr{C}_{q})$, \mycomment{We construct a perturbation $\lambda$ to $\mathscr{C}_{q}$ such that $ \Sigma_{k} (\Sigma_{q}, \mathscr{C}(\lambda)) \prec \Sigma_{k}(\Sigma_{q}, \mathscr{C}_{q})$, and $ \gamma_{k}(\mathscr{C}(\lambda)) \prec \gamma_{k}(\mathscr{C}_{q}) \ \forall \ k$.}
    \begin{align}
        \Sigma_{1}(\Sigma_{q}, \mathscr{C}(\lambda)) &= (A + BK_{0}(\lambda)) \Sigma_{q} (A + BK_{0}(\lambda))\t + DD\t, \nonumber \\
        &=\begin{aligned}
        \Sigma_{1}(&\Sigma_{q}, \mathscr{C}_{q}) + \lambda^{2}_{0}BK_{0}\Sigma_{q}(BK_{0})\t  \\ &\qquad - \lambda_{0} [ BK_{0} \Sigma_{q} (A + BK_{0})\t \\ &\qquad \qquad + (A+BK_{0})\Sigma_{q} (BK_{0})\t ].
        \end{aligned} \label{eq:sigma1_c_cq}
    \end{align}
    From (\ref{eq:sigma1_c_cq}), $\lambda^{2}_{0}BK_{0}\Sigma_{q}(BK_{0})\t \prec \lambda_{0} [ BK_{0} \Sigma_{q} (A + BK_{0})\t + (A+BK_{0})\Sigma_{q} (BK_{0})\t ]$ is a sufficient condition to ensure $ \Sigma_{1} (\Sigma_{q}, \mathscr{C}(\lambda)) \prec \Sigma_{1}(\Sigma_{q}, \mathscr{C}_{q})$. This is done by constructing $\lambda_{0}$ s.t.\ $ \lambda_{\mathrm{max}}(\lambda^{2}_{0}BK_{0}\Sigma_{q}(BK_{0})\t) < \lambda_{\mathrm{min}}(\lambda_{0} [ BK_{0} \Sigma_{q} (A + BK_{0})\t + (A+BK_{0})\Sigma_{q} (BK_{0})\t) ]$. Therefore,
    \begin{align}
        \lambda_{0} < \frac{\lambda_{\mathrm{min}}(BK_{0} \Sigma_{q} (A + BK_{0})\t + (A+BK_{0})\Sigma_{q} (BK_{0})\t)}{ \lambda_{\mathrm{max}}(BK_{0}\Sigma_{q}(BK_{0})\t) }. \label{eq:lambda_0}
    \end{align}
    Repeating the above arguments recursively for $\Sigma_{k}(\Sigma_{q}, \mathscr{C}(\lambda))$ and $\Sigma_{k}(\Sigma_{q}, \mathscr{C}_{q})$, $\Sigma_{k+1}(\Sigma_{q}, \mathscr{C}(\lambda)) = (A + BK_{k}(\lambda)) \Sigma_{k}(\Sigma_{q}, \mathscr{C}(\lambda)) (A + BK_{k}(\lambda))\t + DD\t$. Assuming $\lambda_{k-}$ is such that $ \Sigma_{k}(\Sigma_{q}, \mathscr{C}(\lambda)) \prec \Sigma_{k}(\Sigma_{q}, \mathscr{C}_{q}) $, $\lambda^{2}_{k}BK_{k}\Sigma_{k}(\Sigma_{q}, \mathscr{C}_{q})(BK_{k})\t \preceq \lambda_{k} [ BK_{k} \Sigma_{k}(\Sigma_{q}, \mathscr{C}_{q}) (A + BK_{k})\t + (A+BK_{k})\Sigma_{k}(\Sigma_{q}, \mathscr{C}_{q}) (BK_{k})\t ]$ is a sufficient condition to ensure $ \Sigma_{k+1}(\Sigma_{q}, \mathscr{C}(\lambda)) \preceq \Sigma_{k+1}(\Sigma_{q}, \mathscr{C}_{q})$. Therefore,
    \begin{align}
        \lambda_{k} \leq \frac{\mathscr{N}^{\Sigma}_{k}}{ \lambda_{\mathrm{max}}(BK_{k}\Sigma_{k}(\Sigma_{q}, \mathscr{C}_{q})(BK_{k})\t) } \label{eq:lambda_k}
    \end{align}
    $\forall \ k=0, 1, \cdots, hN-1$, where $ \mathscr{N}^{\Sigma}_{k} \coloneqq \lambda_{\mathrm{min}}(BK_{k} \Sigma_{k}(\Sigma_{q}, \mathscr{C}_{q}) (A + BK_{k})\t + (A+BK_{k})\Sigma_{k}(\Sigma_{q}, \mathscr{C}_{q}) (BK_{k})\t) $.

    We define quantities $\epsilon_{s,k}(\alpha_{x}, \epsilon_{x}; \lambda)$, and $\epsilon_{u,k}(\alpha_{u}, \epsilon_{u}; \lambda)$ as,
    \begin{align}
        &\epsilon_{s,k}(\alpha_{x}, \epsilon_{x}; \lambda) \coloneqq \frac{z_{s,k}}{\lVert \alpha_{x} \rVert \sqrt{ \lambda_{\mathrm{max}}(\gamma_{k}(\mathscr{C}(\lambda))) } } \nonumber \\ &+ \Phi^{-1}(1-\epsilon_{x}) \frac{ \sqrt{\lambda_{\mathrm{min}}( \Sigma_{k}(\Sigma_{q}, \mathscr{C}_{q}))} - \sqrt{\lambda_{\mathrm{max}}( \Sigma_{k}(\Sigma_{q}, \mathscr{C}(\lambda)))} }{\sqrt{ \lambda_{\mathrm{max}}(\gamma_{k}(\mathscr{C}(\lambda))) }}, \label{eq:epsilon_sk_lb} \\
        &\epsilon_{u,k}(\alpha_{u}, \epsilon_{u}; \lambda) \coloneqq \frac{z_{u,k}}{\lvert (1 - \lambda_{k}) \rvert\lVert K\t_{k}(\mathscr{C}_{q})\alpha_{u} \rVert \sqrt{ \lambda_{\mathrm{max}}(\gamma_{k}(\mathscr{C}(\lambda))) } } \nonumber \\ &+ \Phi^{-1}(1-\epsilon_{u}) \frac{ \sqrt{\lambda_{\mathrm{min}}( \Sigma_{k}(\Sigma_{q}, \mathscr{C}_{q}))} - \sqrt{\lambda_{\mathrm{max}}( \Sigma_{k}(\Sigma_{q}, \mathscr{C}(\lambda)))} }{\sqrt{ \lambda_{\mathrm{max}}(\gamma_{k}(\mathscr{C}(\lambda))) }}. \label{eq:epsilon_uk_lb}
    \end{align}
    which are lower bounds on the RHS of (\ref{eq:state_constraint_ineq}) and (\ref{eq:control_constraint_ineq}) respectively. Thus, a sufficient condition for (\ref{eq:state_constraint_ineq}) and (\ref{eq:control_constraint_ineq}) to hold is,
    \begin{align}
        \epsilon \leq \epsilon_{s}(\alpha_{x}, \epsilon_{x}; \lambda) \coloneqq \min_{k} \epsilon_{s,k}(\alpha_{x}, \epsilon_{x}; \lambda), \\
        \epsilon \leq \epsilon_{u}(\alpha_{u}, \epsilon_{u}; \lambda) \coloneqq \min_{k} \epsilon_{u,k}(\alpha_{u}, \epsilon_{u}; \lambda).
    \end{align}
    For specified $\lambda$, $ \epsilon_{s,k}(\alpha_{x}, \epsilon_{x}; \lambda)$ and $ \epsilon_{u,k}(\alpha_{u}, \epsilon_{u}; \lambda) $ are functions of the planning scene $ \{ (\alpha_{x}, \epsilon_{x}), (\alpha_{u}, \epsilon_{u}) \} $ s.t.\ $ \epsilon_{s,k}(\alpha_{x}, \epsilon_{x}; \lambda)$ and $ \epsilon_{u,k}(\alpha_{u}, \epsilon_{u}; \lambda) $ are monotonically decreasing functions of $ \lVert \alpha_{x} \rVert, \lVert \alpha_{u} \rVert \neq 0 $ respectively. From $ \lambda$ specified earlier, the above inequalities are a function of just the planning scene $ \{ (\alpha_{x}, \epsilon_{x}), (\alpha_{u}, \epsilon_{u}) \}$. 
    
    We now derive a condition on $\epsilon$ from the terminal goal reaching constraint and show that there always exists some $\epsilon >0$ such that all the conditions hold by choosing the planning scene in a careful manner. Writing the desired terminal covariance constraint under the perturbations,
    \begin{align}
        \lambda_{\mathrm{min}}( \Sigma^{-}_{\mathrm{max}} ) < \lambda_{\mathrm{max}}( \Sigma_{hN} (\Sigma, \mathscr{C})) \leq \lambda_{\mathrm{min}}(\Sigma_{\mathrm{max}}).
    \end{align}
    We define the function $f: [0,\infty) \rightarrow \mathbb{R}^{+}$ as
    \begin{align}
        f(\epsilon; \lambda) \coloneqq \lambda_{\mathrm{max}}(\Sigma_{hN} (\Sigma, \mathscr{C})) = \lambda_{\mathrm{max}}(\Sigma_{hN} (\Sigma_{q}, \mathscr{C}) + \epsilon \gamma_{k}(\mathscr{C})),
    \end{align}
    where, for a specified $\lambda$, $f(\cdot ; \lambda)$ is a function of $\epsilon$ alone. In particular, $f(\cdot)$ is a continuous, strictly increasing (under mild assumptions on $\mathscr{C}_{q}$ i.e.\ $(A + BK_{k})$ is full rank $\forall \ k$), and a convex function of $\epsilon$. $f(0; \lambda) = \lambda_{\mathrm{max}}( \Sigma_{hN}( \Sigma_{q}, \mathscr{C})) \leq \lambda_{\mathrm{max}}( \Sigma_{hN}( \Sigma_{q}, \mathscr{C}_{q}))$. Also from (\ref{eq:covar_goal_maxcovar}), $ \lambda_{\mathrm{max}} ( \Sigma_{hN} (\Sigma_{q}, \mathscr{C}_{q})) \leq \lambda_{\mathrm{min}}(\Sigma^{-}_{\mathrm{max}})$. Therefore, $ \exists \epsilon^{+}$ s.t.\ $ f(\epsilon^{+}; \lambda) =  \lambda_{\mathrm{min}}(\Sigma^{-}_{\mathrm{max}})$ s.t.\ $ f(\epsilon; \lambda) > \lambda_{\mathrm{min}}(\Sigma^{-}_{\mathrm{max}}) \ \forall \ \epsilon > \epsilon^{+}$.

    From the specified $\lambda$ as before (\ref{eq:lambda_0}) (\ref{eq:lambda_k}), and the existence of $\epsilon^{+}$ (note that the existence of such an $\epsilon^{+}$ and its value are independent of the choice of the planning scene parameters), we choose $ \{ (\alpha_{x}, \epsilon_{x}), (\alpha_{u}, \epsilon_{u}) \} $ s.t.\ $\epsilon_{s,k}(\alpha_{x}, \epsilon_{x}; \lambda), \epsilon_{u,k}(\alpha_{u}, \epsilon_{u}; \lambda) > \epsilon^{+} \ \forall \ k$.
    For such a choice of $\lambda$, $ \{ (\alpha_{x}, \epsilon_{x}), (\alpha_{u}, \epsilon_{u}) \}$, and $ \epsilon \in (\epsilon^{+}, \operatorname{min}(\epsilon_{s}(\alpha_x, \epsilon_x; \lambda), \epsilon_{u}(\alpha_u, \epsilon_u; \lambda)) $,  $ (\mu_{q}, \Sigma(\epsilon)) \xlongrightarrow[kN]{\mathscr{C}(\lambda)} (\mu_{\mathrm{cand}}, \Sigma_{\mathrm{max}})$ and $ \nexists \ \mathscr{C}_{q^{+}} $ s.t.\ $ (\mu_q, \Sigma(\epsilon)) \xlongrightarrow[hN]{\mathscr{C}_{q^{+}}} (\mu_{\mathrm{cand}}, \Sigma^{-}_{\mathrm{max}})$, which completes the proof.
\end{proof}

\begin{proof}[\textbf{Proof of Theorem}~\ref{theorem:coverage}] Note that the $h$-$\operatorname{BRS}$ of a tree of distributions is defined in Section~\ref{sec:hBRS}.
Let $ \mathcal{T}^{(n)} $ be the state of the tree at the end of $n$ iterations of the $\operatorname{ADD}$ procedure. We first show that $ \exists \ \{ ( \alpha_{x}, \beta_{x}, \epsilon_{x} ), ( \alpha_{u}, \beta_{u}, \epsilon_{u} ) \}$ s.t.\ h-$\operatorname{BRS}( \mathcal{T}^{(n+1)} _{\mathrm{MAXCOVAR}}) \supset \text{h-}\operatorname{BRS}( \mathcal{T}^{(n+1)} _{\mathrm{ANY}})$ where $ \mathcal{T}^{(n+1)}_{\mathrm{MAXCOVAR}}$, $ \mathcal{T}^{(n+1)}_{\mathrm{ANY}}$ are trees obtained from a common $\mathcal{T}^{(n)}$ by following one more iteration of the $\operatorname{ADD}$ procedure with the edge constructed according to the $\operatorname{MAXCOVAR}$ algorithm that explicitly maximizes the minimum eigenvalue of the initial covariance  versus any other algorithm for edge construction $\operatorname{ANY}$. The $\operatorname{ANY}$ procedure could for instance sample belief nodes randomly, or add edges that optimize for a performance regularized objective of the minimum eigenvalue of the initial covariance matrix.

    The $(n+1)$-th iteration proceeds by attempting to add a sampled mean $\mu_{q}$ to an existing node $ \nu^{(n)} $ on the tree, where $\nu^{(n)}$ is the node sampled on the $(n+1)$-th iteration. We have,
    \begin{align}
        \Sigma_{\mathrm{max}}, \mathcal{C}_{\mathrm{max}} \longleftarrow \operatorname{MAXCOVAR}(\mu_{q}, ( \mu^{(n)}, \Sigma^{(n)}), N),
    \end{align}
    and,
    \begin{align}
        \Sigma_{\mathrm{any}}, \mathcal{C}_{\mathrm{any}} \longleftarrow \operatorname{ANY}(\mu_{q}, ( \mu^{(n)}, \Sigma^{(n)}), N),
    \end{align}
    s.t.\ $ \lambda_{\mathrm{min}}(\Sigma_{\mathrm{max}}) > \lambda_{\mathrm{min}}(\Sigma_{\mathrm{any}})$. Consider the two nodes added at the $(n+1)$-th iteration: $ \nu^{(n+1)}_{\mathrm{MAXCOVAR}} = (\mu_{q}, \Sigma_{\mathrm{max}})$; and $ \nu^{(n+1)}_{\mathrm{ANY}} = (\mu_{q}, \Sigma_{\mathrm{any}})$ with distance from the goal node as $ d^{(n+1)} $. From Lemma~\ref{lemma:sigma_brs_lemma}, $ \exists \ \{ ( \alpha_{x}, \beta_{x}, \epsilon_{x} ), ( \alpha_{u}, \beta_{u}, \epsilon_{u} ) \}$ s.t.\  $ \text{t-}\operatorname{BRS}(\nu^{(n+1)}_{\mathrm{MAXCOVAR}}) \supset \text{t-}\operatorname{BRS}(\nu^{(n+1)}_{\mathrm{ANY}}) \ \forall \ t \geq 1$. Therefore,
    \begin{align}
        \text{$(h-d^{(n+1)})$-}\operatorname{BRS}(\nu^{(n+1)}_{\mathrm{MAXCOVAR}}) \supset \text{$(h-d^{(n+1)})$-}\operatorname{BRS}(\nu^{(n+1)}_{\mathrm{ANY}}), \label{eq:hbrs_nplus1}
    \end{align}
    $ \forall \ h > d^{(n+1)} $. From Definition~(\ref{def:hbrs_treedist}),
    \begin{align}
        \text{h-}\operatorname{BRS}(\mathcal{T}^{(n+1)}_{\mathrm{MAXCOVAR}}) &= \text{h-}\operatorname{BRS}(\mathcal{T}^{(n)}) \bigcup \nonumber \\ &(h - d^{(n+1)})-\operatorname{BRS}( \nu^{(n+1)}_\mathrm{MAXCOVAR}),
    \end{align}
    and,
    \begin{align*}
        \text{h-}\operatorname{BRS}(\mathcal{T}^{(n+1)}_{\mathrm{ANY}}) &= \text{h-}\operatorname{BRS}(\mathcal{T}^{(n)}) \bigcup \nonumber \\ &(h - d^{(n+1)})-\operatorname{BRS}( \nu^{(n+1)}_\mathrm{ANY}).
    \end{align*}
    From (\ref{eq:hbrs_nplus1}), it follows that there exist planning scenes such that $ \text{h-}\operatorname{BRS}(\mathcal{T}^{(n+1)}_{\mathrm{MAXCOVAR}}) \supset \text{h-}\operatorname{BRS}(\mathcal{T}^{(n+1)}_{\mathrm{ANY}})$ when starting from a common tree state $ \mathcal{T}^{(n)} $. Considering the common tree state as just a singleton set of the goal distribution, $ \mathcal{T}^{(0)} = \{\mathcal{G}\}$. Therefore, there exist planning scenes such that $\text{h-}\operatorname{BRS}(\mathcal{T}^{(1)}_{\mathrm{MAXCOVAR}}) \supset \text{h-}\operatorname{BRS}(\mathcal{T}^{(1)}_{\mathrm{ANY}})$. The result for a general $r$ follows from above, which completes the proof.
    \end{proof}

\small
\renewcommand{\baselinestretch}{0.9}
\bibliographystyle{plain}
\bibliography{references}

\begin{thebibliography}{10}

\bibitem{aggarwal2024sdp_CDC}
Naman Aggarwal and Jonathan~P. How.
\newblock {SDP} synthesis of maximum coverage trees for probabilistic planning under control constraints.
\newblock {\em (to be presented at CDC 2024)}.

\bibitem{aggarwal2024sdp}
Naman Aggarwal and Jonathan~P How.
\newblock \href{https://arxiv.org/abs/2403.14605}{{SDP} Synthesis of Maximum Coverage Trees for Probabilistic Planning under Control Constraints}.
\newblock {\em arXiv preprint arXiv:2403.14605}, 2024.

\bibitem{agha2014firm}
Ali-Akbar Agha-Mohammadi, Suman Chakravorty, and Nancy~M Amato.
\newblock {FIRM:} sampling-based feedback motion-planning under motion uncertainty and imperfect measurements.
\newblock {\em The International Journal of Robotics Research}, 33(2):268--304, 2014.

\bibitem{boyd2004convex}
Stephen~P Boyd and Lieven Vandenberghe.
\newblock {\em Convex optimization}.
\newblock Cambridge university press, 2004.

\bibitem{cvxpy2016}
Steven Diamond and Stephen Boyd.
\newblock {CVXPY:} a python-embedded modeling language for convex optimization.
\newblock {\em The Journal of Machine Learning Research}, 17(1):2909--2913, 2016.

\bibitem{eirm}
Valentin~N Hartmann, Marlin~P Strub, Marc Toussaint, and Jonathan~D Gammell.
\newblock Effort informed roadmaps {(EIRM*)}: Efficient asymptotically optimal multiquery planning by actively reusing validation effort.
\newblock In {\em The International Symposium of Robotics Research}, pages 555--571. Springer, 2022.

\bibitem{horn2012matrix}
Roger~A Horn and Charles~R Johnson.
\newblock {\em Matrix analysis}.
\newblock Cambridge university press, 2012.

\bibitem{prm1996}
Lydia~E Kavraki, Petr Svestka, J-C Latombe, and Mark~H Overmars.
\newblock Probabilistic roadmaps for path planning in high-dimensional configuration spaces.
\newblock {\em IEEE transactions on Robotics and Automation}, 12(4):566--580, 1996.

\bibitem{kurzhanski1996ellipsoidal}
A.~Kurzhanski and I.~Valyi.
\newblock {\em Ellipsoidal Calculus for Estimation and Control}.
\newblock Systems \& Control: Foundations \& Applications. Birkh{\"a}user Boston, 1996.

\bibitem{li2021distributionally}
Bin Li, Yuan Tan, Ai-Guo Wu, and Guang-Ren Duan.
\newblock A distributionally robust optimization based method for stochastic model predictive control.
\newblock {\em IEEE Transactions on Automatic Control}, 67(11):5762--5776, 2021.

\bibitem{dro_lin2022distributionally}
Fengming Lin, Xiaolei Fang, and Zheming Gao.
\newblock Distributionally robust optimization: A review on theory and applications.
\newblock {\em Numerical Algebra, Control and Optimization}, 12(1):159--212, 2022.

\bibitem{liu2022optimal}
Fengjiao Liu, George Rapakoulias, and Panagiotis Tsiotras.
\newblock Optimal covariance steering for discrete-time linear stochastic systems.
\newblock {\em arXiv preprint arXiv:2211.00618}, 2022.

\bibitem{pipx}
Mohamed~Khalid M~Jaffar and Michael Otte.
\newblock {PiP-X:} online feedback motion planning/replanning in dynamic environments using invariant funnels.
\newblock {\em The International Journal of Robotics Research}, page 02783649231209340, 2023.

\bibitem{majumdar2017funnel}
Anirudha Majumdar and Russ Tedrake.
\newblock Funnel libraries for real-time robust feedback motion planning.
\newblock {\em The International Journal of Robotics Research}, 36(8):947--982, 2017.

\bibitem{okamoto2018optimal}
Kazuhide Okamoto, Maxim Goldshtein, and Panagiotis Tsiotras.
\newblock Optimal covariance control for stochastic systems under chance constraints.
\newblock {\em IEEE Control Systems Letters}, 2(2):266--271, 2018.

\bibitem{pan2023distributionally}
Guanru Pan and Timm Faulwasser.
\newblock Distributionally robust uncertainty quantification via data-driven stochastic optimal control.
\newblock {\em IEEE Control Systems Letters}, 2023.

\bibitem{polik2007survey}
Imre P{\'o}lik and Tam{\'a}s Terlaky.
\newblock A survey of the {S}-lemma.
\newblock {\em SIAM review}, 49(3):371--418, 2007.

\bibitem{dro_Rahimian_2022}
Hamed Rahimian and Sanjay Mehrotra.
\newblock Frameworks and results in distributionally robust optimization.
\newblock {\em Open Journal of Mathematical Optimization}, 3:1–85, July 2022.

\bibitem{rapakoulias2023discrete}
George Rapakoulias and Panagiotis Tsiotras.
\newblock Discrete-time optimal covariance steering via semidefinite programming.
\newblock In {\em 2023 62nd IEEE Conference on Decision and Control (CDC)}, pages 1802--1807. IEEE, 2023.

\bibitem{renganathan2023distributionally}
Venkatraman Renganathan, Joshua Pilipovsky, and Panagiotis Tsiotras.
\newblock Distributionally robust covariance steering with optimal risk allocation.
\newblock In {\em 2023 American Control Conference (ACC)}, pages 2607--2614. IEEE, 2023.

\bibitem{shapiro2021distributionally}
Alexander Shapiro.
\newblock Distributionally robust optimal control and {MDP} modeling.
\newblock {\em Operations Research Letters}, 49(5):809--814, 2021.

\bibitem{tedrake2010lqr}
Russ Tedrake, Ian~R Manchester, Mark Tobenkin, and John~W Roberts.
\newblock {LQR}-trees: Feedback motion planning via sums-of-squares verification.
\newblock {\em The International Journal of Robotics Research}, 29(8):1038--1052, 2010.

\bibitem{yakubovich_slemma}
V.~A. Yakubovich.
\newblock {S}-procedure in nonlinear control theory.
\newblock {\em Vestnik Leningrad University, Series 1,13(1):62-77,1971}.

\bibitem{yang2020wasserstein}
Insoon Yang.
\newblock Wasserstein distributionally robust stochastic control: A data-driven approach.
\newblock {\em IEEE Transactions on Automatic Control}, 66(8):3863--3870, 2020.

\bibitem{csbrm-tro}
Dongliang Zheng, Jack Ridderhof, Zhiyuan Zhang, Panagiotis Tsiotras, and Ali-Akbar Agha-mohammadi.
\newblock {CS-BRM:} a probabilistic roadmap for consistent belief space planning with reachability guarantees.
\newblock {\em IEEE Transactions on Robotics}, 2024.

\bibitem{csbrm_tro2024}
Dongliang Zheng, Jack Ridderhof, Zhiyuan Zhang, Panagiotis Tsiotras, and Ali-Akbar Agha-mohammadi.
\newblock {CS-BRM:} a probabilistic roadmap for consistent belief space planning with reachability guarantees.
\newblock {\em IEEE Transactions on Robotics}, 2024.

\end{thebibliography}
\begin{IEEEbiography}[{\includegraphics[width=1in,height=1.25in,clip,keepaspectratio]{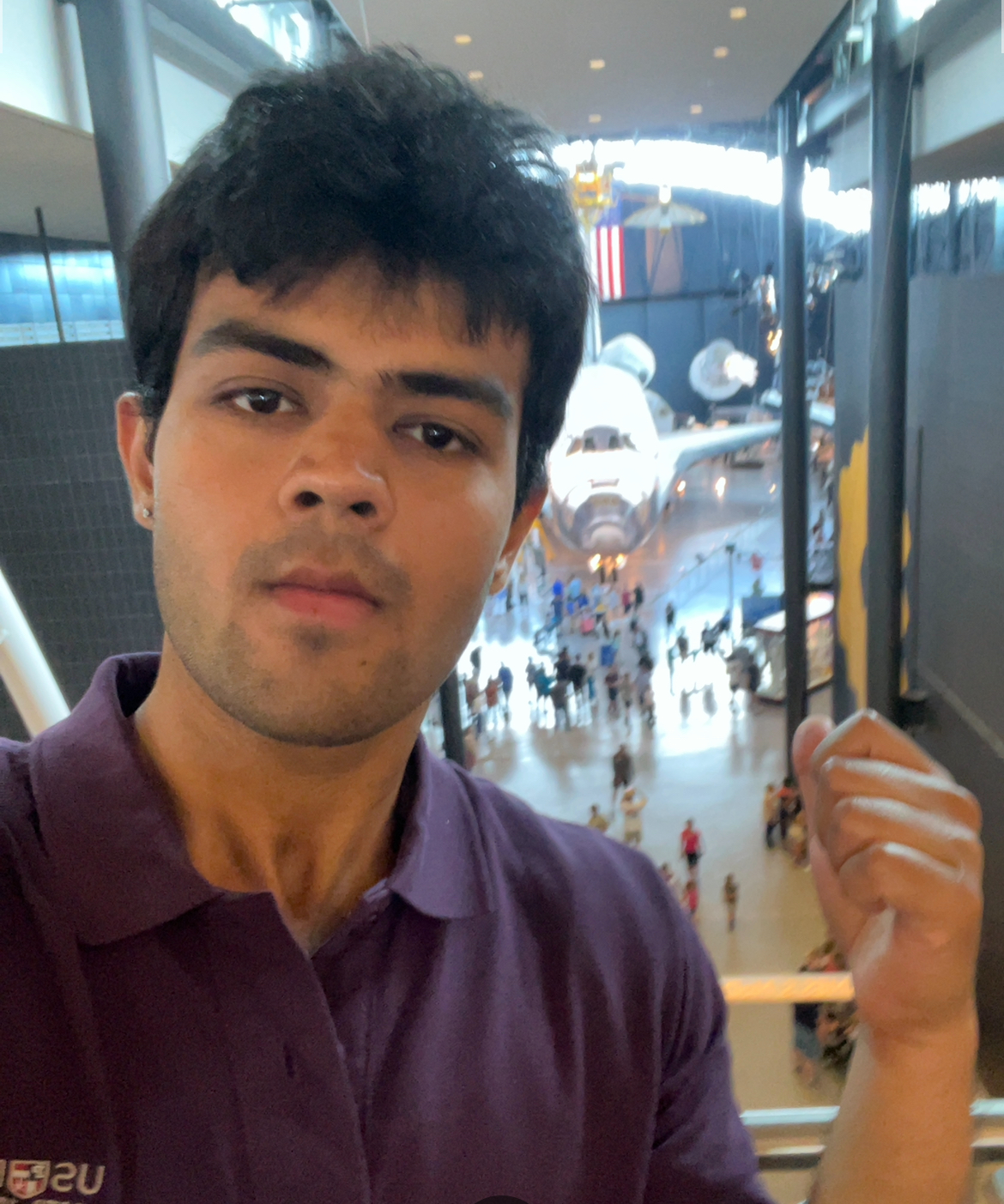}}]%
{Naman Aggarwal} received the Inter-Disciplinary Dual Degree (B.Tech + M.Tech) in Aerospace Engineering and Systems and Control Engineering from the Indian Insitute of Technology Bombay, Mumbai, India in 2021. He is currently  working towards the Ph.D. degree in the Department of Aeronautics and Astronautics and the Laboratory of Information and Decision Systems (LIDS) from Massachusetts Institute of Technology (MIT), Cambridge, MA, USA. 

He is a member of the Aerospace Controls Laboratory, led by Prof. Jonathan How. His current research interests include control theory and optimization with applications at the intersection of learning, games, and multi-agent control.
\end{IEEEbiography}
\begin{IEEEbiography}[{\includegraphics[width=1in,height=1.25in,clip,keepaspectratio]{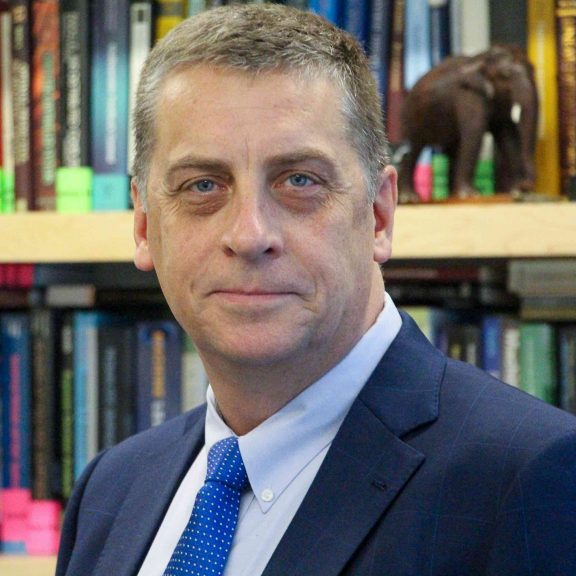}}]%
{Jonathan P. How}
(Fellow, IEEE) received the B.A.Sc. degree in engineering science (aerospace) from the University of Toronto, Toronto, ON, Canada, in 1987, and the S.M. and Ph.D. degrees in aeronautics and astronautics from the Massachusetts Institute of Technology (MIT), Cambridge, MA, USA, in 1990 and 1993, respectively.
In 2000, he joined MIT, where he is currently the Richard C. Maclaurin Professor of Aeronautics and Astronautics. Prior to this, he was an Assistant
Professor with Stanford University, Stanford, CA, USA.
Dr. How was the Recipient of the American Institute of Aeronautics and
Astronautics (AIAA) Best Paper in Conference Awards in 2011, 2012, and 2013, the IROS Best Paper Award on Cognitive Robotics in 2019, the AIAA Intelligent Systems Award in 2020, and the IEEE Control Systems Society Distinguished Member Award in 2020. He was the Editor-in-Chief for IEEE Control Systems Magazine from 2015 to 2019. He is a Fellow of the AIAA. He was elected to the National Academy of Engineering in 2021.
\end{IEEEbiography}
\end{document}